\documentclass{amsart}

\mathcode`l="8000
\begingroup
\makeatletter
\lccode`\~=`\l
\DeclareMathSymbol{\lsb@l}{\mathalpha}{letters}{`l}
\lowercase{\gdef~{\ifnum\the\mathgroup=\m@ne \ell \else \lsb@l \fi}}
\endgroup

\newfont{\authfntsmall}{phvr at 11pt}
\newfont{\eaddfntsmall}{phvr at 9pt}

\usepackage[utf8]{inputenc}
\usepackage[T1]{fontenc}
\usepackage{lmodern}

\usepackage{microtype}

\usepackage{mathtools}
\usepackage{booktabs}
\usepackage{array}
\usepackage[nocompress,nosort]{cite}

\usepackage{float}
\usepackage[noend]{algpseudocode}
\newfloat{algo}{tb}{lop}
\floatname{algo}{Algorithm}
\usepackage{caption}
\def\paragraph#1{\smallskip\noindent{\bf #1.}}

\usepackage{amssymb,amsmath,amsthm}
\usepackage{amsfonts}
\usepackage{xcolor}
\usepackage{url}
\usepackage[plainpages=false,pdfpagelabels,colorlinks=true,citecolor=blue,hypertexnames=false]{hyperref}
\usepackage{paralist}

\newtheorem{thm}{Theorem}
\newtheorem{prop}[thm]{Proposition}
\newtheorem{lem}[thm]{Lemma}
\newtheorem{coro}[thm]{Corollary}

\newtheorem{fact}[thm]{Fact}
\theoremstyle{definition}
\newtheorem{example}[thm]{Example}

\newenvironment{algoenv}[3][\linewidth]{
\begin{minipage}{#1} 
\flushleft
\rule{\textwidth}{.08em}\vspace{-.3\baselineskip}\smallskip
\begin{description}
\item[\rlap{Input}\phantom{Output}] #2
\item[Output] #3
\end{description}
\vspace{-.7\baselineskip}
\rule{\textwidth}{.05em}
\begin{algorithmic}}
{\end{algorithmic}
\vspace{-.5\baselineskip}
\rule{\textwidth}{.08em}
\end{minipage}}

\newcommand{\bigO}{{{O}}}
\newcommand{\softO}{\tilde{\bigO}}
\def\le{\leqslant}
\def\ge{\geqslant}

\def\AA{\mathbb{A}}
\def\KK{\mathbb{K}}
\def\NN{\mathbb{N}}
\def\QQ{\mathbb{Q}}
\def\ZZ{\mathbb{Z}}

\newcommand{\rec}{\operatorname{rec}}
\newcommand{\Newton}{\mathcal{N}}
\newcommand{\Residue}{\operatorname{Residue}}
\newcommand{\Resultant}{\operatorname{Resultant}}

\newcommand{\bideg}{\operatorname{bideg}}
\newcommand{\Diag}{\operatorname{Diag}}
\newcommand{\slope}{\operatorname{ddeg}^-}
\newcommand{\slopesup}{\operatorname{ddeg}^+}
\newcommand{\nsmall}{\operatorname{Nsmall}}
\newcommand{\val}{\operatorname{val}}
\newcommand{\Wr}{\operatorname{Wr}}

\def\gathen#1{{#1}}

\title[Algebraic Diagonals and Walks]{Algebraic Diagonals and Walks: \\ Algorithms, Bounds, Complexity}

\author{Alin Bostan}\address{Inria (France)}
  \email{alin.bostan@inria.fr}
\author{Louis Dumont}
  \address{Inria (France)}
  \email{louis.dumont@inria.fr}
\author{Bruno Salvy}
        \address{Inria (France),
       LIP (U. Lyon, CNRS, ENS Lyon, UCBL)}
       \email{bruno.salvy@inria.fr}

\begin{document}

\begin{abstract}
The diagonal of a multivariate power series $F$ is the univariate power series $\Diag F$ generated by the diagonal terms of $F$. Diagonals form an important class of power series; they occur frequently in number theory, theoretical physics and enumerative combinatorics. We study algorithmic questions related to diagonals in the case where $F$ is the Taylor expansion of a bivariate rational function. It is classical that in this case $\Diag F$ is an algebraic function. We propose an algorithm that computes an annihilating polynomial for $\Diag F$. We give a precise bound on the size of this  polynomial and show that generically, this polynomial is the minimal polynomial and that its size reaches the bound. The algorithm runs in time quasi-linear in this bound, which grows exponentially with the degree of the input rational function. We then address the related problem of enumerating directed lattice walks. The insight given by our study leads to a new method for expanding the generating power series of bridges, excursions and meanders. We show that their first $N$ terms can be computed in quasi-linear complexity in $N$, without first computing a very large polynomial equation.
\end{abstract}
\maketitle

\section{Introduction}

The \emph{diagonal} of a multivariate power series with coefficients~$a_{i_1,\dots,i_k}$ is the univariate power series with coefficients~$a_{i,\dots,i}$. Particularly interesting is the class of diagonals of \emph{rational} power series (ie, Taylor expansions of rational functions). In particular,  diagonals of \emph{bivariate} rational power series are always roots of  nonzero bivariate polynomials (ie, they are algebraic series)~\cite{Polya1921,Furstenberg1967}. 
This property persists for multivariate rational power series, but only in positive characteristic, while the converse inclusion~---~algebraic series being diagonals of rational series~---~always holds~\cite{Furstenberg1967,Safonov87,DenefLipshitz1987}. As far as we are aware, the first occurrence of this result  in the literature  is an article of P\'olya's~\cite{Polya1921}, which
deals with a particular class of bivariate rational functions; the proof uses elementary
complex analysis. Along the lines of P\'olya's approach, Furstenberg~\cite{Furstenberg1967} gave a (sketchy) proof of the general result, over the
field of complex numbers; the same argument has been enhanced later~\cite{HautusKlarner1971},\cite[\S6.3]{Stanley99}. 
Three more different proofs exist: a purely algebraic one that works 
over arbitrary fields of characteristic zero~\cite[Th. 6.1]{Gessel80} (see also~\cite[Th. 6.3.3]{Stanley99}), one based on non-commutative power series~\cite[Prop. 5]{Fliess74}, and a combinatorial proof~\cite[\S3.4.1]{BousquetMelou06} that relies on an encoding of the diagonal using unidimensional walks, seen themselves as words of a non-ambiguous context-free language.
Various other generalizations are known~\cite{Furstenberg1967,Deligne84,Haible97,Pochekutov09}.

\paragraph{Polynomial equations}
Despite the richness of the topic and the fact that most proofs are constructive in essence, we were not able to find in the literature any \emph{explicit} algorithm for computing a bivariate polynomial that cancels the diagonal of a general bivariate rational function. We design in Section~\ref{sec:diagonals} such an algorithm for computing a polynomial equation for the diagonal of an arbitrary bivariate rational function. We show in Proposition~\ref{prop:generic} that generically, the size of the minimal polynomial for the diagonal of a rational function is exponential in the degree of the input and that our algorithm computes it in quasi-optimal complexity (Theorem~\ref{thm:bound diagonals}).
 
 The algorithm has two main steps that may be of independent interest. The first step is the computation of a polynomial equation for the residues of a bivariate rational function. We propose an efficient algorithm for this task, that is a polynomial-time version of Bronstein's algorithm~\cite{Bronstein92}; corresponding size and complexity bounds are given in Theorem~\ref{th:Bronstein}.
The second step is the computation of a polynomial equation for the sums of a fixed number of roots of a given polynomial. We design an additive version of the Platypus algorithm~\cite[\S2.3]{BanderierFlajolet2002} and analyze it in Theorem~\ref{thm:platypus-bound}.

\paragraph{Recurrences}
Since it is also classical that  algebraic series are differentially finite (ie, satisfy linear differential equations with polynomial coefficients),  the coefficients of these bivariate diagonals satisfy  linear recurrences and this leads to an optimal algorithm for the computation of their first terms~\cite{ChudnovskyChudnovsky1986,ChudnovskyChudnovsky1987a,BostanChyzakLecerfSalvySchost2007}. We show however, that computing an annihilating polynomial of the diagonal first is usually not the right approach and that a direct computation of the recurrence~\cite{BoChChLi10} will be more efficient. For completeness, we mention that in more than two variables, diagonals of rational functions are still differentially finite~\cite{Christol85,Lipshitz1988} and currently the most efficient algorithm in that situation is that based on the Griffiths-Dwork method~\cite{BostanLairezSalvy2013,Lairez2014}.

\paragraph{Walks}
Diagonals of rational functions appear  naturally in enumerative combinatorics. In particular, the enumeration of unidimensional walks has been the subject of recent activity, see~\cite{BanderierFlajolet2002} and the references therein.  
Three generating functions of different types of walks are of interest: the generating series $B$ of bridges, $E$ of excursions and $M$ of meanders (these are defined precisely in Section~\ref{sec:walks}). 
The algebraicity of these generating functions is classical as
well, and related to that of bivariate diagonals. Beyond this structural
result, several quantitative and effective results are known. Explicit
formulas give the generating functions in terms of implicit algebraic functions
attached to the set of allowed steps in the cases of
excursions~\cite[\S4]{BoPe00},\cite{Gessel80}, 
bridges and meanders~\cite{BanderierFlajolet2002}. Moreover, 
Bousquet-Mélou gave a tight exponential bound on the degree of the annihilating polynomial in the case of excursions~\cite[\S2.1]{Bousquet2006}, while Banderier and
Flajolet designed an algorithm (called the
\emph{Platypus Algorithm}) computing it~\cite[\S2.3]{BanderierFlajolet2002}. 

Our message for these walks is that again, precomputing a polynomial equation is too costly if one is only interested in the enumeration. 
 Instead, we propose to precompute a differential equation for $B$, that has polynomial size only, to use it for expanding~$B$, and to recover the expansion of~$E$ from that of~$B$. For meanders, we compute a polynomial-size differential equation for $\log M$, from which the expansion of $M$ can be computed efficiently.
Our algorithms have quasi-linear complexity in the precision of the expansion, while keeping the precomputation step in polynomial complexity
(Theorem~\ref{thm:walks}).

\paragraph{Structure of the article}
After a preliminary section on background and notation, we first discuss several special bivariate resultants of broader general interest in Sections~\ref{sec:residues} and~\ref{sec:summation of residues}. Next, we consider 
diagonals, the size of their {minimal} polynomials and an efficient way of computing annihilating polynomials in Section~\ref{sec:diagonals}. 
Finally, we turn to walks in Section~\ref{sec:walks} and show how to compute the coefficients of the generating functions of excursions and of meanders efficiently.

A preliminary version of this article has appeared at the ISSAC'15 conference~\cite{BostanDumontSalvy2015}. In the present version, we give tight bounds in the main results (Theorems~\ref{thm:platypus-bound} and~\ref{thm:bound diagonals}), an improved algorithm for the algebraic residues and more detailed proofs throughout.

\paragraph{Acknowledgments} This work was supported in part by the project FastRelax ANR-14-CE25-0018-01.

\section{Background and Notation}
In this section, 
that might be skipped at first reading,
we introduce notation and technical results that will be used throughout the article.
\subsection{Notation}
In this article, $\KK$ denotes a field of characteristic~0, and $\overline{\KK}$ an algebraic closure of $\KK$. 
We denote by $\KK[x]_n$ the set of polynomials in $\KK[x]$ of degree less than~$n$. Similarly, $\KK(x)_n$ stands for the set of rational functions in $\KK(x)$ with numerator and denominator in $\KK[x]_n$, and $\KK[[x]]_n$ for the set of power series in $\KK[[x]]$ truncated at precision~$n$.

If~$P$ is a polynomial in $\KK[x,y]$, then its degree with respect to $x$ (resp. $y$) is denoted $\deg_xP$ (resp. $\deg_yP$). We take the convention that $\deg 0=-\infty$. The \emph{bidegree} of $P$ is the pair~$\bideg P=(\deg_xP,\deg_yP)$. The notation $\deg$ without any subscript is used for univariate polynomials. Inequalities between bidegrees are component-wise. The set of polynomials in $\KK[x,y]$ of bidegree less than $(n,m)$ is denoted by $\KK[x,y]_{n,m}$, and similarly for more variables.

The \emph{valuation} of a polynomial~$F\in\KK[x]$ or a power series~$F\in\KK[[x]]$ is its smallest exponent with nonzero coefficient. It is denoted $\val F$, with the convention~$\val 0=\infty$.

The \emph{reciprocal} of a polynomial~$P\in\KK[x]$ is the polynomial $\rec(P)=x^{\deg P}P(1/x)$.
If $P=c(x-\alpha_1)\dotsm(x-\alpha_d)$ with $c\neq0$ and $\alpha_i\in\overline{\KK}$ for all $i$, the notation $\Newton(P)$ stands for the generating series of the \emph{Newton sums} of $P$:
\[\Newton(P)=\sum_{n\ge 0}{(\alpha_1^n+\alpha_2^n+\dots+\alpha_d^n)x^n}.\]

A polynomial is called \emph{square-free} when its gcd with its derivative is trivial.
A \emph{square-free decomposition} of a nonzero polynomial $Q\in\AA[y]$, where $\AA=\KK$ or $\KK[x]$, is a factorization~$Q=Q_1^1\dotsm Q_m^m$, with $Q_i\in\AA[y]$ square-free, the $Q_i$'s pairwise coprime and~$\deg_y(Q_m)>0$. The corresponding \emph{square-free part} of~$Q$ is the polynomial~$Q^\star=Q_1\dotsm Q_m$. If $Q$ is square-free then $Q = Q^\star$.

The coefficient of~$x^n$ in a power series $A\in\KK[[x]]$ is denoted $[x^n]A$. 
If $A=\sum_{i=0}^{\infty}{a_ix^i}$, then $A\bmod x^n$ denotes the polynomial $\sum_{i=0}^{n-1}{a_ix^i}$.
The exponential series $\sum_n x^n/n!$ is denoted $\exp(x)$.
The \emph{Hadamard product} of two power series~$A$ and~$B$ is the power series~$A\odot B$ such that $[x^n]A\odot B=[x^n]A\cdot[x^n]B$ for all $n$. 

If $F(x,y) = \sum_{i,j\ge 0}{f_{i,j}x^iy^j}$ is a bivariate power series in $\KK[[x,y]]$, 
the \emph{diagonal} of $F$, denoted $\Diag F$ is the univariate power series in~$\KK[[t]]$ defined by
$\Diag F(t) = \sum_{n\ge 0}{f_{n, n}t^n}.$

\subsection{Complexity Estimates} \label{sec:complexity} 
We recall classical complexity notation and facts for later use. 
Let $\KK$ be again a field of characteristic zero.
Unless otherwise specified, we estimate the cost of our algorithms by counting
arithmetic operations in $\KK$ (denoted ``ops.'') at unit cost. 
The soft-O notation $\softO( \cdot)$ indicates that polylogarithmic factors
are omitted in the complexity estimates (see~\cite[Def. 25.8]{GaGe03} for a precise definition). 
The \emph{arithmetic size} of an element of~$\KK$ is~1. That of a univariate polynomial is its degree plus~1 (ie, we are considering \emph{dense} representations). That of tuples of polynomials is the sum or their sizes, and this defines the size for rational functions and multivariate polynomials. 
We say that an algorithm has quasi-linear complexity if its complexity is $\softO(d)$, where $d$ is the maximal arithmetic size of the input and of the output. In that case, the algorithm is said to be \emph{quasi-optimal}. 

\smallskip\noindent {\bf Univariate operations.} Throughout this article we will use the fact that most operations on polynomials, rational functions and power series in one variable can be performed in quasi-linear time.
Standard references for these questions are the books~\cite{GaGe03} and~\cite{BuClSh97}, as well as~\cite{Schonhage1982}.
The needed results are summarized in Fact~\ref{fact:complexity} below.
\begin{fact}\label{fact:complexity}
The following operations can be performed in $\softO(n)$ ops. in $\KK$:
\begin{compactenum}
\item[(1)] addition, product and differentiation of elements in $\KK[x]_n$, $\KK(x)_n$ and $\KK[[x]]_n$; integration in $\KK[x]_n$ and $\KK[[x]]_n$;
\item[(2)] {extended gcd,}
 square-free decomposition and resultant in $\KK[x]_n$;
\item[(3)] multipoint evaluation in $\KK[x]_n$, $\KK(x)_n$ at $O(n)$ points in $\KK$; interpolation in $\KK[x]_n$ and $\KK(x)_n$
from $n$ (resp. $2n-1$) values at pairwise distinct points in~$\KK$;
\item[(4)] inverse, logarithm, exponential in  $\KK[[x]]_n$ (when defined);
\item[(5)] conversions between $P\in\KK[x]_n$ and $\Newton(P)\bmod x^n \in \KK[x]_n$.
\end{compactenum}
\end{fact}

\smallskip\noindent {\bf Multivariate operations.} 
Basic operations on polynomials, rational functions and power series in several variables are hard questions from the algorithmic point of view. For instance, no general quasi-optimal algorithm is currently known for computing resultants of bivariate polynomials, even though in several important cases such algorithms are available~\cite{BoFlSaSc06}.
Multiplication is the most basic non-trivial operation in this setting. The following result can be proved using Kronecker's substitution; it is quasi-optimal for a fixed number of variables $m=O(1)$. 
For polynomials with more complicated monomial supports, or when the number of variables grows, more sophisticated techniques apply~\cite{CaKaYa89,Pan94,LeSc03,HoSc13}. 
\begin{fact}\label{fact:multiprod}
For fixed~$m$, polynomials in $\KK[x_1, \dots, x_m]_{d_1, \dots, d_m}$ 
and power series in $\KK[[x_1, \dots, x_m]]_{d_1, \dots, d_m}$
can be multiplied using \sloppy 
$\softO(d_1 \dotsm d_m)$ ops.
\end{fact}

A related operation is multipoint evaluation and interpolation. The simplest case is when the evaluation points form an $m$-dimensional tensor product grid $I_1 \times \cdots \times I_m$, where $I_j$ is a set of cardinal $d_j$
; it extends to subgrids of tensor product grids~\cite{HoSc13}.

\begin{fact}\label{fact:multieval}\cite{Pan94}
For fixed $m$, polynomials in $\KK[x_1, \dots, x_m]_{d_1, \dots, d_m}$ can be evaluated and interpolated from values that they take on $d_1 \cdots d_m$ points that form an $m$-dimensional tensor product grid using 
$\softO(d_1 \dotsm d_m)$ ops.
\end{fact}

\noindent Again, the complexity in Fact~\ref{fact:multieval} is quasi-optimal for fixed~$m=O(1)$.

A general (although  non-optimal) technique to deal with more involved operations on multivariable algebraic objects (eg, in $\KK[x,y]$)
is to use (multivariate) evaluation and interpolation on polynomials and to perform operations on the evaluated algebraic objects using Facts~\ref{fact:complexity}--\ref{fact:multieval}. To put this strategy in practice, the size of the output needs to be well controlled. We illustrate this philosophy on the example of resultant computation, based on the following easy variation of~\cite[Thm.~6.22]{GaGe03}.
\begin{fact} \label{fact:resultant}
  Let $P(x,y)$ and $Q(x,y)$ be bivariate polynomials of respective bidegrees $(d_x^P, d_y^P)$ and $(d_x^Q, d_y^Q)$.
Then, $$\deg\Resultant_y(P(x,y), Q(x,y)) \le d_x^Pd_y^Q+d_x^Qd_y^P,$$
and this is an equality whenever one of $d_x^Q$ or $d_x^P$ is zero.
\end{fact}

\begin{lem} \label{algo resultant}
	Let $P$ and $Q$ be polynomials in $\KK[x_1, \dots, x_m, y]_{d_1, \dots, d_m, d}$. 
	Then $R = \Resultant_y(P, Q)$ belongs to $\KK[x_1, \dots, x_m]_{D_1, \dots, D_m}$, where $D_i = 1+2(d-1)(d_i-1)$. Moreover, the coefficients of $R$ can be computed using $\softO(2^m d_1 \cdots d_m d^{m+1})$ ops. in~$\KK$. \end{lem}

\begin{proof}
	The degrees estimates follow from Fact~\ref{fact:resultant}. 
To compute $R$, we use an evaluation-interpolation scheme:  $P$ and $Q$ are evaluated at $D=D_1 \cdots D_m$ points $(x_1, \dots, x_m)$ forming an $m$ dimensional tensor product grid; $D$ univariate resultants in $\KK[y]_d$ are computed; $R$ is recovered by interpolation. By Fact~\ref{fact:multieval}, the evaluation and interpolation steps are performed in $\softO(m D)$ ops. The second one has cost $\softO(d D)$. Using the inequality $D \le 2^m d_1 \cdots d_m d^m $ concludes the proof.
\end{proof}

	We conclude this section by recalling  two complexity results on bivariate polynomials and rational functions; for proofs, see~\cite{Lecerf08} and~\cite{BoChChLi10}.
\begin{fact}\label{fact:sqfree}
	\begin{compactenum}
		\item[(1)] A square-free decomposition of polynomials in\\$\KK[x,y]_{d_x, d_y}$ can be computed using $\softO(d_x^2  d_y)$ ops.
		\item[(2)] If $P,Q\in\KK[x,y]$ are non-zero coprime polynomials such that $\bideg(P)<\bideg(Q)$ and $Q$ is primitive wrt $y$, then a minimal telescoper for $P/Q$ of degree $\bigO(d_xd_y^\star d_y)$ and order at most $d_y^\star$ can be computed using $\softO(d_xd_y^2d_y^{\star3})$ ops, where $(d_x,d_y)=\bideg(Q)$ and $d_y^\star$ is the degree in $y$ of any square-free part of $Q$.
	\end{compactenum}
\end{fact}
Recall that a minimal telescoper for $P/Q$ is a differential operator $L\in\KK[x]\langle \partial_x\rangle$ of minimal order such that $L\cdot(P/Q)=\partial_y(g)$ with $g\in\KK[x,y]$.

\section{Polynomials for Residues}\label{sec:residues}
\subsection{Algorithm}
We are interested in a polynomial that vanishes at some or all of the residues of a given rational function. It is a classical result in symbolic integration that in the case of simple poles, there is a resultant formula for such a polynomial, first introduced by Rothstein~\cite{Rothstein1976} and Trager~\cite{Trager1976}. This was later generalized by Bronstein~\cite{Bronstein92} to accommodate multiple poles as well. However, as mentioned by Bronstein, the complexity of his method grows exponentially with the multiplicity of the poles. Instead, we develop in this section an algorithm with polynomial complexity.

Let $f = P/Q$ be a nonzero element in $\KK(y)$, where $P, Q$ are two coprime
polynomials in $\KK[y]$. 
Let also $\hat{Q}$ be a divisor of $Q$ such that $\hat{Q}$ and $Q/\hat{Q}$ are coprime. In our context, $\hat{Q}$ represents the subset of the roots of $Q$ at which we want to compute an annihilating polynomial of the residues.
Let $Q_1Q_2^2\cdots Q_m^m$ be a square-free decomposition of~$\hat{Q}$.
For $i\in\{1,\dots,m\}$, if $\alpha$ is a root of~$Q_i$ in an algebraic extension of~$\KK$, then it is simple and the residue of~$f$ at~$\alpha$ is the coefficient of~$t^{-1}$ in the Laurent expansion of $f(\alpha+t)$ at~$t=0$. Consider the polynomial $V_i(y,t)= (Q_i(y+t)-Q_i(y))/t$. Since $\alpha$ is a simple root of $Q_i$, $V_i$ satisfies $V_i(\alpha,t) = Q_i(\alpha+t)/t$ and $V_i(\alpha,0)=Q_i'(\alpha)\neq 0$. Therefore, the rational function $g$ defined by $g(y,t) = f(y+t)Q_i^i(y+t)/V_i^i(y,t)$ satisfies $g(\alpha,t)=f(\alpha+t)\cdot t^i$ and has the advantage of being regular at $t=0$. The residue of $f$ at $\alpha$ may hence be computed as the evaluation at $y=\alpha$ of $[t^{i-1}]g(y,t)$. If this coefficient is denoted~$S_{i-1}(y)=A_i(y)/B_i(y)$, with polynomials $A_i$ and $B_i$,  the residue at~$\alpha$ is thus a root of $\Resultant_y(A_i-zB_i,Q_i)$. When the multiplicity of the pole $m=1$, this is exactly the Rothstein-Trager resultant. 
This computation leads to Algorithm~\ref{algo:Bronstein}, which avoids the exponential blowup of the complexity that would follow from a symbolic precomputation of the Bronstein resultants.

\begin{algo}	
	Algorithm \textbf{AlgebraicResidues}$(P/Q, \hat{Q})$
	
	\begin{algoenv}{Three polynomials $P$, $Q$ and $\hat{Q}$ a divisor of~$Q$ in $\KK[y]$ such that $\hat{Q}$ and $Q/\hat{Q}$ are coprime ($\hat{Q}$ can be $Q$)}{A polynomial in $\KK[z]$ canceling the residues of $P/Q$ at the roots of~$\hat{Q}$}
		\State Compute $Q_1Q_2^2\dotsm Q_m^m$ a square-free decomposition of~$\hat{Q}$;
		\For{$i\gets1$ to $m$}
		\If{$\deg_yQ_i=0$} \ $R_i\gets1$
		\Else
		\State $U_i(y)\gets Q(y)/Q_i^i(y)$;
		\State $V_i(y,t)\gets (Q_i(y+t)-Q_i(y))/t$;
		\State Expand $\frac{P(y+t)}{U_i(y+t)V_i^i(y,t)}=S_0+\dotsb+S_{i-1}t^{i-1}+O(t^{i})$;
    \State Write $S_{i-1}$ as $A_i(y)/B_i(y)$ with $A_i$ and $B_i$ coprime polynomials;
		\State $R_i(z)\gets \Resultant_y(A_i-zB_i,Q_i)$;
		\EndIf
		\EndFor
		\State\Return 		
		$R_1R_2\dotsm R_m$ 
	\end{algoenv}
	\caption{Polynomial canceling the residues}\label{algo:Bronstein}
\end{algo}

\begin{example}\label{ex:Bronstein}
Let $d\ge 0$ be an integer, and let $G_d(x,y)\in\QQ(x)[y]$ be the rational function $y^d/(y-y^2-x)^{d+1}$. 
The poles have order~$d+1$. In this example, the algorithm can be performed by hand for arbitrary~$d$:
a square-free decomposition has~$m=d+1$ and~$Q_{m}=y-y^2-x$, the other $Q_i$'s being~1. Then~$V_m=1-2y-t$ and the next step is to expand
\[\frac{(y+t)^d}{(1-2y-t)^{d+1}}=\frac{(y+t)^d}{(1-2y)^{d+1}\left(1-\frac{t}{1-2y}\right)^{d+1}}.\]
Expanding the binomial series gives the coefficient of~$t^d$ as~$\frac{A_m}{B_m}$, with
\[A_m=\sum_{k=0}^d{\binom{d}{k}\binom{d+k}{k}y^k(1-2y)^{d-k}},\quad B_m=(1-2y)^{2d+1}.\]
The residues are then cancelled by~$R_m = \Resultant_y(A_m-zB_m,Q_{m})$, namely by
\begin{equation}\label{eq:pol-diag-ex}
R_m = (1-4x)^{2d+1} z^2 - \left( \sum_{k=0}^{\lfloor d/2 \rfloor} \binom{d}{2k}\binom{2k}{k} x^k \right)^{\!\!2}.
\end{equation}
(Equality~\eqref{eq:pol-diag-ex} is a consequence of the identity~\footnote{Both sides of the identity satisfy
	$(2y-1)^2 (d+1) u_d-(2d+3)u_{d+1}+(d+2)u_{d+2}=0, u_0=u_1=1.$} 
$A_m = \sum_{k=0}^{\lfloor d/2 \rfloor} \binom{d}{2k}\binom{2k}{k} (y-y^2)^k$,
which implies
$A_m \bmod Q_m = 
\sum_{k=0}^{\lfloor d/2 \rfloor} \binom{d}{2k}\binom{2k}{k} x^k$, while 
$B_m \bmod Q_m = (1-4x)^d(1-2y)$.) 
\end{example}
In our applications, as in the previous example, the polynomials $P$ and~$Q$ have coefficients that are themselves polynomials in another variable~$x$. The rest of this section is devoted to the proof of the following.

\begin{thm}\label{th:Bronstein}
	Let $P(x,y)/Q(x,y)\in\KK(x,y)_{d_x+1,d_y+1}$. Let $\hat{Q}$ be a divisor of~$Q$, $\hat{Q}^\star$ be a square-free part of it wrt $y$, and denote by $m$ the number of factors in the square-free decompositions of $\hat{Q}$. Let~$(d_x^\star,d_y^\star)$ be  bounds on the bidegree of ${Q}^\star$. Then the polynomial computed by Algorithm~\ref{algo:Bronstein} annihilates the residues of~$P/Q$ at the roots of $\hat{Q}$, has degree in~$z$ bounded by~$\deg_y\hat{Q}^\star$ and degree in~$x$ bounded by 
\[2d_x^\star(d_y+1)+2(d_y^\star-1)d_x-2d_x^\star d_y^\star.\]
It can be computed in $\bigO(m^2d_x^\star d_y^\star(m^2 + {d_y^\star}^2))$ operations in~$\KK$.
\end{thm}
Note that rewriting the bound under the equivalent form
\[2d_xd_y-2(d_x-d_x^\star)(d_y-d_y^\star+1)\]
shows that the degree in~$x$ is bounded by~$2d_xd_y$, independently of the multiplicities.
The complexity is also bounded independently of the multiplicities by
$\bigO(d_x^\star d_y^\star d_y^4)$.

\subsection{Bounds}
By Fact~\ref{fact:resultant}, the resultant~$R_i$ has degree in~$z$ exactly $\deg Q_i$ so that the degree in~$z$ of the result is bounded by $\deg_yQ_1+\dots+\deg_yQ_m=d_y^\star$.

The degree in~$x$ is the sum of the degrees in~$x$ of all the~$R_i$'s. In order to derive a bound on the degree of~$R_i$ using Fact~\ref{fact:resultant}, we first consider the degrees in~$x$ and~$y$ of~$A_i$ and~$B_i$. The important point is that these degrees do not depend so much on~$Q$ as on its square-free part. In order to quantify this precisely, we first focus on power series expansion of a special type about which we state a few useful lemmas.

For a polynomial~$Q\in\KK[x]$ and a real number~$\alpha$, we denote by $\mathcal{E}_\alpha(Q)$ the subset of~$\KK(x)[[t]]$ formed of power series that can be written
\[c_0+c_1\frac{t}{Q}+\dotsb+c_n\frac{t^n}{Q^n}+\dotsb,\]
with $c_n\in\KK[x]$ and $\deg c_n\le n\alpha$, for all $n$ (recall that $\deg0=-\infty$, which makes it convenient to allow negative~$\alpha$). This notation extends to the case when $x$ is a tuple of variables, with $\alpha$ replaced by a tuple of real numbers. The main properties of $\mathcal{E}_\alpha(Q)$ are summarized as follows. 
\begin{lem} \label{growth}
	Let $Q, R\in\KK[x]$, $\alpha, \beta\in\mathbb{R}$ and $f\in\KK[[t]]$.
	\begin{compactenum}
		\item[(1)] The set $\mathcal{E}_\alpha(Q)$ is a subring of~$\KK(x)[[t]]$;		
		\item[(2)] Let $S\in\mathcal{E}_\alpha(Q)$ with $S(0)=0$, then $f(S)\in\mathcal{E}_\alpha(Q)$;
		\item[(3)] The products obey
		\[\mathcal{E}_\alpha(Q)\cdot\mathcal{E}_\beta(R) \subset \mathcal{E}_{\max(\alpha+\deg R,\,\beta+\deg Q)}(QR).\]
	\end{compactenum} 
\end{lem}
\begin{proof} 
For~{\em(3)}, if $A = \sum_n{a_nt^n/Q^n}$ and $B = \sum_n{b_nt^n/R^n}$ belong respectively to~$\mathcal{E}_\alpha(Q)$ and $\mathcal{E}_\beta(R)$, then the $n$th coefficient of their product is a sum of terms of the form $a_i(x)Q^{n-i} b_{n-i}(x)R^i/(QR)^n$. Therefore, the degree of the numerator is  bounded by 
$i(\alpha+\deg R)+(n-i)(\beta+\deg Q)$, whence {\em(3)} is proved. Property~{\em(1)} is proved similarly, the $n$th coefficient of the product being a sum of terms $a_i(x)b_{n-i}(x)t^n/Q^n$.
In Property~{\em(2)}, the condition on~$S(0)$ makes~$f(S)$ well-defined. The result then follows from~{\em(1)}.
\end{proof}

\begin{coro}\label{coro:invpol}
Let~$Q\in\KK[x,t]$ be such that~$Q(0,0)\neq0$. Let $Q^\star$ be a square-free part of $Q$ and $\delta(Q^\star)$ its total degree in $(x,t)$. Then 
\[\frac{1}{Q(x,t)}\in\frac{1}{Q(x,0)}\mathcal{E}_{\min(\deg_x(Q^\star),\delta(Q^\star)-1)}(Q^\star(x,0)).\]
\end{coro}

\begin{proof}
For all~$i$, the coefficient of~$t^i$ in~$Q$ has degree at most~$\min(\deg_x(Q),\delta(Q)-i)$. Thus $R:=(Q(x,t)-Q(x,0))/Q(x,0)\in\mathcal{E}_{\min(\deg_x(Q),\delta(Q)-1)}(Q(x,0))$. Writing $Q(x,t)=Q(x,0)(1+R)$ and using Part~{\em(2)} of Lemma~\ref{growth} with~$f=1/(1+y)$ then gives the result when~$Q$ is square-free. Using $f=1/(1+y)^i$ gives the result for a pure power by Part~{\em(1)} of the lemma. The general case then follows from 
Part~{\em(3)} by induction on the number of parts in the square-free decomposition of~$Q$, using additivity of degree and total degree.
\end{proof}

Now, we turn to the fraction
$F_i:={P(x,y+t)}/U_i(x,y+t)/V_i(x,y,t)^i$, with $U_i(x,y)={Q(x,y)}/{Q_i(x,y)^i}$ and $V_i(x,y,t)=({Q_i(x,y+t)-Q_i(x,y)})/{t}$. 
We use bidegrees with respect to $(x,y)$ and observe that
\[\bideg U_i^\star=\bideg Q^\star-\bideg Q_i,\quad \bideg V_i\le\bideg Q_i-(0,1).\]
The total degrees in $(y,t)$ behave similarly: that of~$U_i^\star(x,y+t)$ is $\deg_yQ-\deg_yQ_i$, while that of $V_i(x,y,t)$ is $\deg_yQ_i-1$. 
Corollary~\ref{coro:invpol} gives
\begin{align}
\frac{1}{U_i(x,y+t)}&\in\frac{1}{U_i(x,y)}\mathcal{E}_{\bideg Q^\star-\bideg Q_i-(0,1)}(U_i^\star(x,y)),\label{eq:degoneoverui}\\
\frac{1}{V_i(x,y,t)^i}&\in\frac{1}{V_i(x,y,0)^i}\mathcal{E}_{\bideg Q_i-(0,2)}(V_i(x,y,0)).\label{eq:degoneovervii}
\end{align}
From there,  Part~{\em(3)} of Lemma~\ref{growth} shows that the product of these series belongs to
\[\frac{1}{U_i(x,y)V_i(x,y,0)^i}\mathcal{E}_{\bideg Q^\star-(0,2)}(U_i^\star(x,y)V_i(x,y,0)).\]
Thus
the coefficient $S_{i-1}$ of $t^{i-1}$ in the power series expansion of $F_i$ can be written as~$A_i/B_i$ with 
\[B_i=U_i(x,y)V_i(x,y,0)^iU_i^\star(x,y)^{i-1}V_i(x,y,0)^{i-1},\]
and finally
\begin{equation}\label{eq:degAB}
\begin{split}
\bideg A_i&\le\bideg P+(i-1)\bideg Q^\star-2(i-1)(0,1),\\
\bideg B_i&\le\bideg Q+(i-1)\bideg Q^\star-(2i-1)(0,1),
\end{split}
\end{equation}
whence
\[\bideg(A_i-zB_i)\le \max(\bideg P,\bideg Q-(0,1))+(i-1)(\bideg Q^\star-(0,2)).\]
Fact~4 can now be exploited, leading to a bound on the degree of the resultant:
\begin{multline*}
\deg_xR_i\le \deg_yQ_i\left(\max(\deg_xP,\deg_xQ)+(i-1)\deg_xQ^\star\right)\\
+\deg_xQ_i\left(\max(\deg_yP,\deg_yQ-1)+(i-1)(\deg_yQ^\star-2)\right).
\end{multline*}

Next, we sum over the indices~$i$ corresponding to factors of~$\hat{Q}$. This leads to the following bound for the degree in~$x$ of the result
\begin{multline*}
\deg_y\hat{Q}^\star\max(\deg_xP,\deg_xQ)+(\deg_y\hat{Q}-\deg_y\hat{Q}^\star)\deg_xQ^\star\\
+\deg_x\hat{Q}^\star\max(\deg_yP,\deg_yQ-1)+(\deg_x\hat{Q}-\deg_x\hat{Q}^\star)(\deg_yQ^\star-2).
\end{multline*}
This bound being an increasing function of each of the degrees that appear, it is itself upper bounded by replacing any of those degrees by an upper bound.

In the context of Theorem~\ref{th:Bronstein}, the bidegrees of $P$, $Q$ and $\hat{Q}$ are bounded by $(d_x,d_y)$, while those of~$Q^\star$ and $\hat{Q}^\star$ are bounded by $(d_x^\star,d_y^\star)$. This leads to the bound
\[
d_y^\star d_x+(d_y-d_y^\star)d_x^\star+d_x^\star d_y+(d_x-d_x^\star)(d_y^\star-2), 
\]
which rewrites as the bound in the Theorem and completes that part of the proof.

\subsection{Complexity}
By Fact~\ref{fact:sqfree}, a square-free decomposition of $\hat{Q}$ can be computed using $\softO(d_x^2  d_y)$ ops.
We now focus on the computations performed inside the $i$th iteration of the loop and write $(d_x^{(i)},d_y^{(i)})$ for the bidegree of~$Q_i$.
Computing $U_i$ requires an exact division of polynomials of bidegrees at most $(d_x, d_y)$; this division can be performed by evaluation-interpolation in $\softO(d_x d_y)$ ops. Similarly, the trivariate polynomial $V_i$ can be computed by evaluation-interpolation wrt $(x,y)$ in time $\softO(d_x^{(i)}(d_y^{(i)})^2)$. By Eq.~\eqref{eq:degAB}, both $A_i(x,y)$ and $B_i(x,y)$ have bidegrees at most $(D_x^{(i)},D_y^{(i)})$, where $D_x^{(i)} = d_x+id_x^\star$ and $D_y^{(i)} =  d_y+id_y^\star$. They can be computed by evaluation-interpolation in 
$\softO(i D_x^{(i)} D_y^{(i)})$ ops. Finally, the resultant $R_i(x,z)$ has bidegree at most $(d_x^{(i)}D_y^{(i)} + d_y^{(i)}D_x^{(i)}, d_y^{(i)})$, and since the degree in~$y$ of $A_i-zB_i$ and~$Q_i$ is at most~$D_y^{(i)}$, it can be computed by evaluation-interpolation in $\softO((d_x^{(i)}D_y^{(i)} + d_y^{(i)}D_x^{(i)})d_y^{(i)}D_y^{(i)})$ ops by Lemma~\ref{algo resultant}. The total cost of the loop is thus $\softO(L)$, where
\[ L = \sum_{i=1}^m \left( (i+(d_y^{(i)})^2) D_x^{(i)} D_y^{(i)} + d_x^{(i)} d_y^{(i)} (D_y^{(i)})^2 \right).\]
Using the (crude) 
bounds $D_x^{(i)} \le D_x^{(m)}$, $D_y^{(i)} \le D_y^{(m)}$, $\sum_{i=1}^m (d_y^{(i)})^2 \le {d_y^\star}^2$ and $\sum_{i=1}^m d_x^{(i)} d_y^{(i)} \le d_x^\star d_y^\star$ shows that $L$ is bounded by
\[ 
D_x^{(m)} D_y^{(m)} \sum_{i=1}^m (i+(d_y^{(i)})^2) + (D_y^{(m)})^2 \sum_{i=1}^m  d_x^{(i)} d_y^{(i)} \le 
D_x^{(m)} D_y^{(m)} (m^2+{d_y^\star}^2) + (D_y^{(m)})^2   d_x^\star d_y^\star,
\]
which, by using the inequalities $D_x^{(m)} \le 2m d_x^\star$ and $D_y^{(m)} \le 2m d_y^\star$, is seen to belong to $O(m^2d_x^\star d_y^\star(m^2 + {d_y^\star}^2))$, as was to be proved.
This completes the proof of the theorem.

\noindent {\bf Remark.} Note that one could also use Hermite reduction combined with the usual Rothstein-Trager resultant in order to compute a polynomial $\tilde{R}(x,z)$ that annihilates the residues. Indeed, Hermite reduction computes an auxiliary rational function that admits the same residues as the input, while only having simple poles. A close inspection of this approach provides the same bound $d_y^\star$ for the degree in $y$ of $\tilde{R}(x,z)$, but a less tight bound for its degree in $x$, namely worse by a factor of $d_y^\star$. The complexity of this alternative approach appears to be $\softO(d_x d_y (d_y + {d_y^\star}^3))$ (using results from~\cite{BoChChLi10}), to be compared with the complexity bound from Theorem~\ref{th:Bronstein}.

\section{Sums of roots of a polynomial} \label{sec:summation of residues}
\subsection{Algorithm}
Given a polynomial~$P\in\KK[y]$ of degree~$d$ with coefficients in a field~$\KK$ of characteristic~0, let 
$\alpha_1,\dots,\alpha_d$ be its (not necessarily distinct) roots in the algebraic closure of~$\KK$. For any positive integer $c\le d $, the polynomial of degree $\binom{d}{c}$ defined by
\begin{equation}\label{eq:defSigma}
\Sigma_c P = \prod_{i_1 < \cdots < i_c}{\left(y-(\alpha_{i_1}+\alpha_{i_2}+\cdots+\alpha_{i_c})\right)}
\end{equation}
has coefficients in~$\KK$. This section discusses the computation of~$\Sigma_cP$ summarized in Algorithm~\ref{algo:Sigma_c}, which can be seen as an additive analogue of the \emph{Platypus algorithm} of Banderier and Flajolet~\cite{BanderierFlajolet2002}.

\begin{algo}
	Algorithm \textbf{PureComposedSum}$(P, c)$
	
	\begin{algoenv}{A polynomial $P$ of degree $d$ in $\KK[y]$, a positive integer $c\le d$}{The polynomial $\Sigma_c P$ from Eq.~\eqref{eq:defSigma}}
		\State $D \gets\binom{d}{c}$
		\State $\Newton(P) \gets \rec(P')/\rec(P)\bmod y^{D+1}$
		\State $S \gets \Newton(P)\odot\exp(y) \bmod y^{D+1}$
		\State $F\gets\exp\left(\sum_{n=1}^c(-1)^{n-1}\frac{S(n y)}{n}z^n\right)\bmod (y^{D+1},z^{c+1})$
		\State $\Newton(\Sigma_c P)\gets ([z^c]F)\odot\sum{n!y^n}\bmod y^{D+1}$
		\State\Return 			
$\rec\left(\exp\left(\int\frac{D-\Newton(\Sigma_c P)}{y}\,\mathrm{d}y\right)\bmod y^{D+1}\right)$
	\end{algoenv}
	\caption{Polynomial canceling the sums of $c$ roots}\label{algo:Sigma_c}
\end{algo}

We recall two classical formulas for the generating function of the Newton sums (see, eg, \cite[\S2]{BoFlSaSc06}), the second one being valid for monic $P$ only:
\begin{equation}\label{eq:NewtonSums}
\Newton(P) = \frac{\rec(P')}{\rec(P)},\qquad
\rec(P) = \exp \left(\int{\frac{d-\Newton(P)}{y}\,\mathrm{d}y}\right).
\end{equation}
Truncating these formulas at order~$d+1$ makes~$\Newton(P)$ a representation of the polynomial~$P$ (up to normalization), since both conversions above can be performed quasi-optimally by Newton iteration~\cite{Schonhage1982, Pan2000a, BoFlSaSc06}.
The key for Algorithm~\ref{algo:Sigma_c} is the following variant of~\cite[\S2.3]{BanderierFlajolet2002}.
\begin{prop}
	Let~$P\in\KK[y]$ be a polynomial of degree~$d$,
let $\Newton(P)$ denote the generating series of its Newton sums and let $S$ be the series $\Newton(P) \odot \exp(y)$.
Let $\Psi_c$ be the polynomial in $\KK[t_1, \dots, t_c]$ defined by
	\[\Psi_c(t_1,\dots, t_c) = [z^c]\exp\left(\sum_{n \ge 1}{(-1)^{n-1}t_n\frac{z^n}{n}} \right).\]
	Then the following equality holds
	\[\Newton(\Sigma_c P)\odot \exp(y) = \Psi_c(S(y), S(2 y), \dots, S(c y)).\]
\end{prop}
\begin{proof}
	By construction, the series $S$ is 
	\[S(y)=\sum_{n\ge 0}{(\alpha_1^n+\alpha_2^n+\cdots+\alpha_d^n)\frac{y^n}{n!}} = \sum_{i=1}^{d}\exp(\alpha_i y).\]
	When applied to the polynomial~$\Sigma_c P$, this becomes
	\begin{align*}
	\Newton(\Sigma_c P)\odot \exp(y) &= \sum_{i_1 < \cdots < i_c}{\exp \left({(\alpha_{i_1} + \alpha_{i_2}+\cdots+\alpha_{i_c})y}\right)}\\
	&= [z^c]\prod_{i=1}^{d}{\left(1+z \exp(\alpha_i y)\right)}.
	\end{align*}
This expression rewrites:
	\begin{align*}
[z^c]\exp\left(\sum_{i=1}^{d}\log(1+z \exp({\alpha_i y}))\right)&= [z^c]\exp\left(\sum_{i=1}^{d}\sum_{m\ge 1}{(-1)^{m-1}\exp({\alpha_i m y})\frac{z^m}{m}}\right)\\
	&= [z^c]\exp\left(\sum_{m\ge 1}{(-1)^{m-1}S(m y)\frac{z^m}{m}}\right),
	\end{align*}
and the last expression equals $\Psi_c(S(y), S(2 y), \dots, S(c y))$.
\end{proof}

The correctness of Algorithm~\ref{algo:Sigma_c} follows from observing that the truncation orders $D+1$ in $y$ and $c+1$ in $z$ of the power series involved in the algorithm are sufficient to enable the reconstruction of~$\Sigma_cP$ from its first Newton sums by~\eqref{eq:NewtonSums}.
 
We will be interested in the case where~$P$ is a polynomial in~$\KK[x,y]$. Then, the coefficients of $\Sigma_cP$ wrt $y$ may have denominators. We analyze the structure of the coefficients of $\Sigma_c P$ as elementary symmetric functions of the roots of $P$ in order to compute bounds on the bidegree of the polynomial obtained by clearing out these denominators.
The rest of this section proves the following result.

\begin{thm} \label{thm:platypus-bound}
	Let $P \in \KK[x,y]_{d_x+1,d_y+1}$, and let $c\le d_y$ be a positive integer. Let $a\in\KK[x]$ denote the leading coefficient of $P$ wrt $y$ and let $\Sigma_c P$ be defined as in~Eq.~\eqref{eq:defSigma}. 
	We also denote
	\[D_x := \binom{d_y-1}{c-1}, \qquad D_y := \binom{d_y}{c}.\]
	Then $a^{D_x} \cdot \Sigma_c P$ is a polynomial in $\KK[x,y]$ that cancels all sums $\alpha_{i_1}+\cdots+\alpha_{i_c}$ of $c$ roots $\alpha_i(x)$ of $P$, with $i_1 < \cdots < i_c$, and satisfies
	\[\deg_x (a^{D_x}\cdot \Sigma_c P) \le d_xD_x, \quad \deg_y (a^{D_x}\cdot\Sigma_c P) = D_y.\]
Moreover, this polynomial can be computed in $\softO(c d_x D_xD_y)$ ops.
\end{thm}
\noindent These bounds are sharp. Experiments
suggest that for generic $P$ of bidegree $(d_x, d_y)$ the minimal polynomial of
$\alpha_{i_1}+\cdots+\alpha_{i_c}$ has bidegree precisely $(d_xD_x,D_y)$. Similarly, the complexity result is quasi-optimal up to a factor of $c$ only. 

\subsection{Bounds}
We start with the following effective version of a very classical result on symmetric functions~\cite[Theorem 6.21]{Yap1999}.
\begin{lem}\label{lem:symmetric functions}
Let $\alpha_1, \dots, \alpha_n$ be indeterminates, and $\sigma_1,\dots,\sigma_n$ be the associated elementary symmetric functions. Let $P\in\KK[\alpha_1,\dots,\alpha_n]$ be a symmetric polynomial satisfying
\[\deg_{\alpha_i}P\le d\ \text{for all}\ 1\le i\le n\]
Then $P$ can be expressed as a polynomial in $\sigma_1,\dots,\sigma_n$ of total degree at most $d$.
\end{lem}

\begin{proof}
This is a consequence of the form of the matrix of the change of bases from the elementary symmetric functions to the monomial symmetric functions as described for instance in the proof of \cite[Theorem 7.4.4]{Stanley99}.
Since $P$ is symmetric and has degree at most $d$ with respect to each variable, it can be written as a linear combination of monomial symmetric functions of the form $\sum_{i_1<\dots<i_k}{\alpha_{i_1}^{\lambda_1}\dotsm \alpha_{i_k}^{\lambda_k}}$, where $\lambda_i\le d$ for all $i$. These monomial symmetric functions can in turn be written as linear combinations of elementary symmetric functions of the form $\sigma_{\mu_1}\dotsm\sigma_{\mu_l}$ where $l\le d$, which is exactly the result of the lemma.
\end{proof}

For the proof of the bounds in Theorem~\ref{thm:platypus-bound},
we write \[P=a(x)y^{d_y}+\sum_{i=0}^{d_y-1}{a_i(x)y^i}=a(x)\prod_{i=1}^{d_y}(y-\alpha_i(x)).\]
Let $\sigma_1(x),\dots,\sigma_{d_y}(x)$ denote the elementary symmetric functions of the $\alpha_i$'s.
	Then, the elementary symmetric functions of the roots $\alpha_{i_1}+\dots+\alpha_{i_c}$ of $\Sigma_c P$ have degree $\binom{d_y-1}{c-1}$ in each $\alpha_i$.
	Therefore, by Lemma~\ref{lem:symmetric functions}, the coefficients of $\Sigma_c P$ are polynomials of total degree at most $\binom{d_y-1}{c-1}$ in $\sigma_1,\dots,\sigma_{d_y}$. From there, the bound on $\deg_x (a^{D_x}\cdot\Sigma_c P)$ is immediately derived from the classical relations $(-1)^i\sigma_i=a_{d_y-i}/a$.
	
 \subsection{Complexity}
 The computation is performed by evaluation and interpolation at $1+d_xD_x$ values of~$x$. By Fact~\ref{fact:complexity}, at each of these values, the computation of the truncated series expansions~$\mathcal{N}({P})$ and $S$ in~$\KK[[y]]_{1+D_y}$ have complexity~$\softO(D_y)$; so do the computations of $\mathcal{N}(\Sigma_cP)$ and the last step; the most expensive step is the computation of~$F$, which costs~$\softO(cD_y)$ ops. in~$\KK$. Since this is executed $O(d_xD_x)$ times, the total cost is~$\softO(cd_xD_xD_y)$.

\section{Diagonals}\label{sec:diagonals}

In this section we turn to our main topic, namely the computation of annihilating polynomials for diagonals of bivariate rational functions. The algorithm relies on a classical expression of the diagonal as a sum of residues (see Lemma~\ref{lem:diagonal-as-residues}), and on the results of Sections~\ref{sec:residues} and~\ref{sec:summation of residues}. The conclusions of the analysis of Algorithm~\ref{algo:diagonal} can be found in Theorem~\ref{thm:bound diagonals} and Proposition \ref{prop:generic}.

\subsection{Algebraic equations for diagonals} \label{subsec:Algebraic equations for diagonals}

Let $F(x,y)=\sum_{i,j\ge 0}{a_{i,j}x^iy^j}$ be a rational function in $\KK(x,y)$, whose denominator does not vanish at $(0,0)$.
Then the diagonal of $F$ is defined as $\Diag F(t) =\sum_{i\ge 0}{a_{i,i}t^i}$. A first basic, but very important, remark is that
\[\Diag F(t)=[y^{-1}]\frac{1}{y}F\!\left(\frac{t}{y}, y\right).\]
When $\KK=\mathbb{C}$, this coefficient can be viewed as a Cauchy integral and computed by the residue formula~\cite{Furstenberg1967}. For general $\KK$ (of characteristic $0$), we proceed similarly with a purely algebraic approach, adapted from~\cite[Theorem~6.1]{Gessel80}. (The reader who is not interested in the general proof may also skip directly to Lemma~\ref{lem:diagonal-as-residues}.)
The starting point is the partial fraction decomposition of $G(t,y) := \frac{1}{y}F(\frac{t}{y}, y)$ considered as a rational function in $\KK(t)(y)$:
\begin{equation} \label{eq:partial fractions of G}
	G(t,y) = \sum_{i=1}^n\sum_{j=1}^{m_i}{f_{i,j}(t,y)},
\end{equation}
where
\begin{equation*} 
	f_{i,j}(t,y) = \frac{r_{i,j}(t)}{(y-y_i(t))^j}, \qquad 1\le i\le n,\quad 1\le j\le m_i.
\end{equation*}
In particular, $r_{i,1}(t)$ is the residue of $G$ at $y_i(t)$ for all $i\in\{1,2,\dots,n\}$. By Puiseux's theorem, there exists $N\in\NN^\star$ such that the $y_i$'s and $r_{i,j}$'s all lie in the field $\overline{\KK}((t^{1/N}))$. In order to apply the operator $[y^{-1}]$ on both sides of Equation~\eqref{eq:partial fractions of G}, it is necessary to find a ring where both the equality and the operator $[y^{-1}]$ make sense. We are going to check that $\AA = \KK((y))((t^{1/N}))$ and $[y^{-1}]$ computed coefficient-wise are suitable for this.

First, as a rational function, it is immediate that $G(t,y)$ belongs to $\AA$. In order to expand the right-hand side, we consider each term separately and  distinguish between the cases $\val_t(y_i)\le0$ and $\val_t(y_i)>0$.
If $\val_t(y_i)\le 0$,  $f_{i,j}$ can be written as follows:
\[f_{i,j}  =\frac{r_{i,j}}{(-y_i)^{j}}\cdot\frac{1}{\left(1-y/y_i\right)^{j}} 
   = \frac{r_{i,j}}{(-y_i)^{j}}\sum_{k\ge 0}{\binom{-{j}}{k}\frac{y^k}{y_i^k}}\in\overline{\KK}((t^{1/N}))[[y]].
\]
Since $\val_t(1/y_i)\ge 0$, the series $f_{i,j}/r_{i,j}$ actually belongs to $\overline{\KK}[[t^{1/N}]][[y]]\cong\overline{\KK}[[y]][[t^{1/N}]]$. Hence $f_{i,j}\in\overline{\KK}[[y]]((t^{1/N}))\subset \AA$, and in particular $[y^{-1}]f_{i,j}=0$.
On the other hand, if $\val_t(y_i)>0$ then $f_{i,j}$ can be expanded directly in $\AA$ as:
\[
f_{i,j}  =\frac{r_{i,j}}{y^{j}}\cdot\frac{1}{\left(1-y_i/y\right)^{j}} 
   =\frac{r_{i,j}}{y^{j}}\sum_{k\ge 0}{\binom{-{j}}{k}\frac{y_i^k}{y^k}}.\]
Since $y_i/y \in \AA$ and $\val_t(y_i/y)>0$, this last quantity is the sum of a convergent series (in the sense of formal Laurent series) of elements of $\AA$, hence belongs to $\AA$. In this case we obtain $[y^{-1}]f_{i,j}=r_{i,1}$.

We have everything we need to apply $[y^{-1}]$ on both sides of Equation~\eqref{eq:partial fractions of G}, leading to the generalization to any base field of characteristic $0$ of Furstenberg's classical result~\cite[\S 2]{Furstenberg1967}.
\begin{lem}\label{lem:diagonal-as-residues} If $F(x,y)$ is a rational function in~$\KK(x,y)$ whose denominator does not vanish at~$(0,0)$, then 
\begin{equation}\label{eq:diagonal-as-residues}
\Diag F(t) = \sum_{\substack{y(t)\in\mathcal{P}\\ \val_t(y(t))>0}}{\Residue\!\left(\frac{1}{y}F\!\left(\frac{t}{y}, y\right),y=y(t)\right)},
\end{equation}
where $\mathcal{P}$ is the set of poles of $\frac{1}{y}F(\frac{t}{y}, y)$. 
\end{lem}
The poles $y(t)\in\mathcal{P}$ such that $\val_t(y(t))>0$ are called the \emph{small branches} of~$Q$ and we denote their number by~$\nsmall(Q)$.

Since the elements of~$\mathcal{P}$ are algebraic and finite in number and residues are obtained by series expansion, which entails only rational operations, it follows that the diagonal is algebraic too. 
Combining the algorithms of the previous section gives Algorithm~\ref{algo:diagonal} that produces a polynomial equation for $\Diag F$.

\begin{example}
Let $d\ge 0$ be an integer, and let $F_d(x,y)$ be the rational function $1/(1-x-y)^{d+1}$. The diagonal of $F_d$ is equal to
\[ \sum_{n \ge 0} \binom{2n+d}{n}\binom{n+d}{d}t^n.\]
By the previous argument, it is an algebraic series, which is the sum of the residues of the rational function~$G_d$ of Example~\ref{ex:Bronstein} over its small branches (with $x$ replaced by $t$). In this case, the denominator is~$y-t-y^2$. It has one solution tending to~0 with~$t$; the other one tends to~$1$. Thus the diagonal is canceled by the quadratic polynomial~\eqref{eq:pol-diag-ex}.
\end{example}

\begin{algo}
	Algorithm \textbf{AlgebraicDiagonal}(A/B)
	
	\begin{algoenv}{Two polynomials $A$ and $B$ in $\KK(x,y)$, with $B(0,0)\neq 0$}{A polynomial $\Phi\in\KK[t,\Delta]$ such that $\Phi(t, \Diag A/B)=0$}
		\State $P, Q,\alpha \gets y^{\slope(A)}A(\frac{t}{y},y)$, $y^{\slope(B)}B(\frac{t}{y},y)$, $\slope(B)-\slope(A)-1$
		\If{$\alpha< 0$}
		\State $r\gets \textbf{AlgebraicResidues}(y^\alpha P/Q, y)$
		\EndIf
		\State $R \gets \textbf{AlgebraicResidues}(y^\alpha P/Q, Q)$
		\State $c \gets$ number of small branches of $Q$
		\State $\Phi(t,\Delta) \gets \operatorname{numer}(\textbf{PureComposedSum}(R, c))$
		\If{$\alpha<0$}
		\State $\Phi(t,\Delta)\gets\operatorname{numer}(\Phi(t,\Delta-r))$
		\EndIf
		\State\Return $\Phi(t,\Delta)$
	\end{algoenv}
	\caption{Polynomial canceling the diagonal of a rational function. The notation $\operatorname{ddeg}$ is defined in Eq.~\eqref{eq:def-slope}; $\operatorname{numer}$ denotes the numerator of the irreducible form of a fraction.\label{algo:diagonal}}
\end{algo}

\begin{example}For an integer~$d>0$, we consider the rational function
\[F_d(x,y)=\frac{x^{d-1}}{1-x^{d}-y^{d+1}},\]
of bidegree~$(d,d+1)$. The first step of Algorithm~\ref{algo:diagonal} produces
\[G_d(t,y)=y^\alpha\frac{P}{Q}=\frac{t^{d-1}}{y^{d}-t^{d}-y^{2d+1}},\]
of bidegree~$(d,2d+1)$, 
whose denominator is irreducible with $d$ small branches. From there, Algorithm~\ref{algo:diagonal} computes a polynomial~$\Phi_d$ annihilating $\Diag F_d$, which is experimentally irreducible and whose bidegrees for $d=1,2,3,4$ are $(2,3)$, $(18,10)$, $(120,35)$, $(700,126)$.
From these values, it is easy to conjecture that the bidegree is given by~\[\left(d(d+1)\binom{2d-1}{d-1},\binom{2d+1}{d}\right),\] of exponential growth in the bidegree of $F_d$. In general, these bidegrees do not grow faster than in this example. In Theorem~\ref{thm:bound diagonals} below, we prove bounds that are barely larger than the values above.
\end{example}

\smallskip \noindent \textbf{Sloped Diagonals.} If~$p$ and $q$ are relatively prime positive integers and $F(x,y)=\sum_{i,j\ge0}{f_{i,j}x^iy^j}$, then the \emph{sloped diagonal} of $F$, $\Diag_{p,q}F(t)$ is $\sum_{n\ge0}{f_{pn,qn}t^n}$. 
Direct manipulations show that \[\Diag_{p,q}F(t^{pq})=\Diag(F(x^q,y^p))(t),\] so that our bounds and algorithm apply almost directly to these more general diagonals.

\subsection{Degree Bounds and Complexity} 
The rest of this section is devoted to the derivation of bounds on the complexity of Algorithm~\ref{algo:diagonal} and on the size of the polynomial it computes, which are given in Theorem~\ref{thm:bound diagonals}.

\smallskip \noindent \textbf{Degrees.}
A bound on the bidegree of $\Phi$ will be obtained from the bounds successively given by Theorems~\ref{th:Bronstein} and~\ref{thm:platypus-bound}.

In order to follow the impact of the change of variables in the first step, we define the \emph{lower diagonal degree} and \emph{upper diagonal degree} of a polynomial $P(x,y) =\sum_{i,j}{a_{i,j}x^iy^j}$ respectively as the integers
\begin{equation}\label{eq:def-slope}
\begin{split}
\slope(P) = \sup\left\{i-j\ |\ a_{i,j}\neq 0\right\} \\ \slopesup(P) = \sup\left\{j-i\ |\ a_{i,j}\neq 0\right\}
\end{split}
\end{equation}

We collect the properties of interest in the following.
\begin{lem}\label{lemma:slope}
	For any~$P$ and~$Q$ in~$\KK[x,y]$,
	\begin{compactenum}
		\item[(1)] $\slope(P)\le\deg_xP$ and $\slopesup(P)\le\deg_yP$;
		\item[(2)] $\operatorname{ddeg}^\pm(PQ)=\operatorname{ddeg}^\pm(P)+\operatorname{ddeg}^\pm(Q)$;
		\item[(3)] there exists a polynomial $\tilde{P}\in\KK[x,y]$, such that\\
		$P(x/y,y)=y^{-\slope(P)}\tilde{P}(x,y)$, with $\tilde{P}(x,0)\neq0$ and
		\[\bideg(\tilde{P})= (\deg_x P, \slope(P)+\slopesup(P));\]
		\item[(4)] $\bideg((\tilde{P})^\star)=(\deg_xP^\star,\slope(P^\star)+\slopesup(P^\star))$.
	\end{compactenum}
\end{lem}
\begin{proof}
	Part~{(\em1)} is immediate. The quantities $\slope(P)$ and $\slopesup(P)$ are nothing else than $-\val_yP(x/y,y)$ and $\deg_yP(x/y,y)$, which makes Parts~{(\em2)} and~{(\em3)} clear too. From there, we get the identity $\widetilde{PQ}=\tilde{P}\tilde{Q}$ for arbitrary~$P$ and~$Q$, whence $(\tilde{P})^\star=\widetilde{P^\star}$ and Part~{(\em4)} is a consequence of Parts~{(\em1)} and~{(\em3)}.
\end{proof}
	
Thus, starting with a rational function $F=A/B\in\KK(x,y)$, with $(d_x, d_y)$ a bound on the bidegrees of $A$ and $B$, and $(d_x^\star,d_y^\star)$ a bound on the bidegree of a square-free part $B^\star$ of $B$, the first step of the algorithm constructs $G(t,y)=y^{\alpha}\frac{P}{Q}$, with polynomials~$P$ and~$Q$ and
\begin{gather}
\alpha = \slope(B) -\slope(A) - 1 \label{eq:defalpha}\\
\notag\bideg P \le  (d_x, \slope(A) + \slopesup(A)),\quad 
\bideg Q \le  (d_x, \slope(B) + \slopesup(B)),\\
\notag \bideg Q^\star =(d_x^\star,\slope(B^\star)+\slopesup(B^\star)).
\end{gather}
We first explain how to compute the number $c$ of small branches of $Q$.

\smallskip \noindent \textbf{Small branches.} It is classical that for a polynomial~$P=\sum{a_{i,j}x^iy^j}\in\KK[x,y]$, the number of its solutions tending to~0 can be read off its Newton polygon (see, e.g.~\cite{Walker1978}). This polygon is the lower convex hull of the union of $(i,j)+\mathbb{N}^2$ for $(i,j)$ such that~$a_{i,j}\neq0$. The number of solutions tending to~0 is given by the minimal $y$-coordinate of its leftmost points. Since the number of small branches counts only distinct solutions, it is thus given by 
\begin{equation}\label{eq:Nsmall}
\nsmall(P)=\nsmall(P^\star)=\val_y([x^{\val_xP^\star}]P^\star).
\end{equation}

The change of variables~$x\mapsto x/y$ changes the coordinates of the point corresponding to~$a_{i,j}$ into~$(i,j-i)$. This transformation maps the vertices of the original Newton polygon to the vertices of the Newton polygon of the Laurent polynomial $P(x/y,y)$. Multiplying by~$y^{\slope(P)}$ yields a polynomial and shifts the Newton polygon up by~$\slope(P)$, thus
\[\nsmall\left(y^{\slope(P)}P(x/y,y)\right)=\nsmall(P^\star)+\slope(P^\star).\]

The number of small branches of the polynomial $Q$ constructed above is then given by
\begin{equation}\label{eq:nb_small}
c:=\nsmall{(B^\star)}+\slope(B^\star).
\end{equation}

\smallskip\noindent\textbf{Degree in $\Delta$.}
At this point, there is a slight difference between the cases $\alpha\ge 0$ and $\alpha<0$. Indeed, in the latter case we have to take the additional small branch at $0$ into account. To do this, we denote by $r$ the residue of $G$ at $0$. Since $r$ is rational, we may compute a polynomial $R$ that vanishes only on the residues at non-zero small branches of the denominator of $G$. If $\tilde{\Phi}(t, \Delta)$ is the polynomial produced by applying Algorithm~\ref{algo:Sigma_c} to $(R, c)$, then the polynomial $\Phi(t, \Delta)=\tilde{\Phi}(t, \Delta-r)$ cancels $\Diag F$.
Thus we apply Algorithm~\ref{algo:Bronstein} to $((y^\alpha P)/Q, Q)$ if $\alpha\ge 0$, and to $(P/(y^{-\alpha}Q), Q)$ otherwise. By Theorem~\ref{th:Bronstein}, in both cases we obtain a polynomial $R$ of degree $D_y$, with
\begin{equation}\label{eq:Dy}
D_y:=\slope(B^\star)+\slopesup(B^\star),
\end{equation}
and applying Algorithm~\ref{algo:Sigma_c} gives a polynomial $\Phi$ with $\deg_\Delta\Phi = \binom{D_y}{c}$.

\smallskip\noindent\textbf{Degree in $t$.} To bound the degree of $\Phi$ in $t$, we can neglect our optimization and apply Algorithm~\ref{algo:Bronstein} to $(y^\alpha P, Q, Q)$ or $(P, y^{-\alpha}Q, y^{-\alpha}Q)$ depending on wether $\alpha\ge 0$ or $\alpha<0$. Indeed, the polynomial $\Phi$ obtained this way is clearly a multiple of the one computed by the algorithm. By Theorem~\ref{th:Bronstein}, since the bidegrees of $P$, $y^\alpha P$, $Q$ and $y^{-\alpha}Q$ are all bounded by $(d_x, d_x+d_y+1)$, we compute a polynomial $R$ of degree bounded by $D_x$, where
\begin{equation}
D_x :=2d_x(d_x+d_y+1)+d_x\label{eq:Dx}
-2(d_x-d_x^\star)(d_x-d_x^\star + d_y-d_y^\star+1).
\end{equation}
Applying Theorem~\ref{thm:platypus-bound} to $(R, c)$ or $(R, c+1)$ depending on the sign of $\alpha$ yields in both cases $\deg_t{\Phi}\le D_x\binom{D_y}{c}$.

\smallskip \noindent \textbf{Complexity.} We now analyze the cost of Algorithm~\ref{algo:diagonal}. The computation of $P$ and $Q$ does not require any arithmetic operation. Next, the computation of~$R$ and $r$ takes $\softO((d_x+d_y)^6)$ ops. (see the comment after Theorem~\ref{th:Bronstein}). The number of small branches is obtained with no arithmetic operation from a square-free decomposition computed in Algorithm~\ref{algo:Bronstein}. The bounds of the discussion above and Theorem~\ref{thm:platypus-bound} show that Algorithm~\ref{algo:Sigma_c} uses~$\softO(cD_x\binom{D_y}{c}^2)$ ops. Finally, if a translation of the variable is needed, it can be performed by evaluation-interpolation in $\softO(D_x\binom{D_y}{c}^2)$ ops. (One may as well evaluate and interpolate wrt $x$ and apply better algorithms for univariate translation \cite[\S 5]{BoFlSaSc06}.)

We summarize all the results of this section in the following theorem.
\begin{thm} \label{thm:bound diagonals}
	Let $F=A/B$ be a rational function in $\KK(x,y)$ with $B(0,0)\neq0$. Let $(d_x, d_y)$ (resp. $(d_x^\star,d_y^\star)$) be a bound on the bidegrees of $A$ and~$B$ (resp. a square-free part of~$B$). Let~$D_x,D_y,c$ be defined as in Eqs.~(\ref{eq:Dx},\ref{eq:Dy},\ref{eq:nb_small}).
	Then there exists a polynomial $\Phi\in\KK[t,\Delta]$ such that $\Phi(t, \Diag F(t))=0$ and 
	\[\deg_\Delta \Phi=\binom{D_y}{c}, \qquad \deg_t \Phi \le D_x\binom{D_y}{c}.\]
	Algorithm~\ref{algo:diagonal} computes it in~$\softO \left(cD_x\binom{D_y}{c}^2+(d_x+d_y)^6 \right)$ ops.
\end{thm}
A general bound on $\bideg \Phi$ depending only on a bound $(d,d)$ on the bidegree of the input can be deduced from the above as
\[\bideg\Phi\le(d(4d+3),1)\times\binom{2d}{d}.\]

\subsection{Optimization}\label{subsec:Optimization}
Assume that the denominator of~$F(x/y)/y$ is already partially factored as $Q(y)=\tilde{Q}(y)\prod_{i=1}^k{(y-y_i(x))^{\alpha_i}}$, where the $y_i$'s are~$k$ distinct \emph{rational} branches among the~$c$ small branches of~$Q$. Then their corresponding (rational) residues $r_i$ contribute to the diagonal. 
The special case where $k = 1$ and $y_1=0$ is exactly the situation that occurred in the discussion on $\deg_\Delta \Phi$ before Theorem~\ref{thm:bound diagonals}, when $\alpha<0$. The trick that we used extends directly to the general case: it suffices to apply Algorithm~\ref{algo:Bronstein} to $\tilde{Q}$, Algorithm~\ref{algo:Sigma_c} with $c-k$ roots, and $\Phi$ is then recovered through a change of variable.

\subsection{Generic case}\label{subsec:Generic case}
The bounds from Theorem~\ref{thm:bound diagonals} on the bidegree of $\Phi$ are slightly pessimistic wrt the variable $t$, but generically tight wrt the variable~$\Delta$, as will be proved in Proposition~\ref{prop:generic} below. We first need a lemma.

\begin{lem}\label{lemma:galois groups}
	Let $\KK$ be a field of characteristic $0$, and $P\in \KK[y]$ be a polynomial of degree $d$, with Galois group $\mathfrak{S}_d$ over $\KK$. Assume that the roots $\alpha_1,\dots,\alpha_d$ of $P$ are algebraically independent over $\mathbb{Q}$.
	Then, for any $c\le d$, the degree $\binom{d}{c}$ polynomial $\Sigma_cP$ is irreducible in $\KK[y]$.
\end{lem}

\begin{proof}
	Since $\Sigma = \alpha_1 + \cdots+\alpha_c$ is a root of $\Sigma_c P$, it suffices to prove that $\KK(\Sigma)$ has degree $\binom{d}{c}$ over $\KK$.
	The $\alpha_i$'s being algebraically independent, any permutation $\sigma \in \mathfrak{S}_d$ of all the $\alpha_i$'s that leaves~$\Sigma$ unchanged has to preserve the sets $\left\{\alpha_1,\dots,\alpha_c\right\}$ and $\left\{\alpha_{c+1},\dots,\alpha_d\right\}$. Conversely, any such permutation induces an automorphism of $\KK(\alpha_1,\dots,\alpha_d)$ that leaves $\Sigma$ invariant. In other words, the Galois group of $\KK(\alpha_1,\dots,\alpha_d)$ over $\KK(\Sigma)$ is equal to $\mathfrak{S}_c\times\mathfrak{S}_{d-c}.$
	It follows that $\KK(\alpha_1,\dots,\alpha_d)$ has degree~$c!(d-c)!$ over~$\KK(\Sigma)$ and degree~$d!$ over~$\KK$, so that~$\KK(\Sigma)$ has degree~$\binom{d}{c}$ over~$\KK$.
\end{proof}

\begin{prop}\label{prop:generic}
	Let $A$ be a polynomial in $\QQ[x,y]$. Let $d_x, d_y$ be non-negative integers, $s^-\le d_x$, $s^+\le d_y$, and
	\[B(x,y)=\sum_{i=0}^{s^-}{b^{(x)}_ix^i}+\sum_{j=1}^{s^+}{b^{(y)}_jy^j}+\sum_{\substack{i\le d_x,j\le d_y \\ -s^-\le j-i \le s^+}}{b_{i,j}x^iy^j}\in\QQ[(b^{(x)}_i),(b^{(y)}_j),x,y],\]
	where the $b_i^{(x)}$ and $b_j^{(y)}$ are indeterminates and $b_{i,j}\in\QQ$. 
	
	Then the polynomial computed by Algorithm~\ref{algo:diagonal} with input $A/B$ is irreducible of degree $\binom{s^-+s^+}{s^-}$ over $\KK=\QQ[(b^{(x)}_i),(b^{(y)}_j),x,y]$.
\end{prop}

\begin{proof}
	First apply the change of variables to obtain $G=y^\alpha P/Q$, with \[Q(x,y) = \sum_{i=0}^{s^-}{b^{(x)}_ix^iy^{s^--i}}+\sum_{j=1}^{s^+}{b^{(y)}_jy^{s^-+j}}+\sum_{i, j}{b_{i,j}x^iy^{s^--i+j}}.\]
	Denote $d=s^-+s^+$. Then, the polynomial $Q(1,y)$ has the form $\sum_{j\le d}{t_jy^j}$ where each of the $t_j$'s is the sum of one of the indeterminates and rational constants. This implies that the $t_j$'s are algebraically independent over $\mathbb{Q}$. Therefore, $Q(1,y)$ has Galois group $\mathfrak{S}_{d}$ over $\QQ(t_0,\dots,t_d)$ and its roots are algebraically independent over $\mathbb{Q}$~\cite[\S 57]{Waerden49}. This property lifts to $Q(x,y)$ \cite[\S 61]{Waerden49}, which thus has Galois group $\mathfrak{S}_d$ and algebraically independent roots, denoted $y_1,\dots,y_d$.
	
	Now define the polynomial $R(x,y) = \prod_{i}{(y-\tilde{P}(x,y_i)/\partial_yQ(x,y_i))}$, where $\tilde{P}=y^\alpha P$ if $\alpha\ge 0$ and $\tilde{P}=P$ otherwise. Since $Q$ has simple roots, this is exactly the polynomial that is computed by Algorithm~\ref{algo:Bronstein}. The family $\left\{P(x,y_i)/\partial_yQ(x,y_i)\right\}$ is algebraically independent, since any algebraic relation between them would induce one for the $y_i$'s by clearing out denominators. In particular, the natural morphism $\operatorname{Gal}(Q/\KK)=\mathfrak{S}_{d}\rightarrow \operatorname{Gal}(R/\KK)$ is injective, whence an isomorphism. (Here, $\operatorname{Gal}(P/\KK)$ denotes the Galois group of $P\in\KK[y]$ over~$\KK$.) Since an immediate investigation of the Newton polygon of $Q$ shows that it has $s^-$ small branches, we conclude using Lemma~\ref{lemma:galois groups} and the fact that the translation of the variable doesn't change the irreducible character of $\Phi$.
\end{proof}

Proposition~\ref{prop:generic} should be viewed as an optimality result. Indeed, for a generic rational function $A/B$ as in the proposition, we have $B=B^\star$, $\slope(B)=s^-$, $\slopesup(B)=s^+$ and $B$ has $s^-$ small branches. This implies that the bound of Theorem~\ref{thm:bound diagonals} for $\deg_\Delta\Phi$ is optimal in this (generic) case.

If one believes that random examples should behave like the generic case, then the proposition means that the polynomial computed by Algorithm~\ref{algo:diagonal} will be irreducible most of the time.

As an example, we consider the special case of Proposition~\ref{prop:generic} where $s^-=s^+=d_x=d_y = d$. In this case, $\deg_\Delta\Phi$ is $\binom{2d}{d}$. We compare this to the following experiment on random examples.
\begin{example} 
	We consider a rational function $F(x,y)=1/{B(x,y)}$, where $B(x,y)$ is a dense polynomial of bidegree $(d,d)$ chosen at random. For~$d=1,2,3,4$, algorithm \textbf{AlgebraicDiagonal}($F$) produces \emph{irreducible} outputs with bidegrees
$(2,2)$, $(16,6)$, $(108, 20)$, $(640, 70) $, that are matched by the formulas 
\begin{equation}\label{eq:bound_generic}
\left(2d^2\binom{2d-2}{d-1},\binom{2d}{d}\right),
\end{equation}
so that the bound on $\deg_\Delta\Phi$ is tight in this case and the irreducibility of the output shows that Theorem~\ref{thm:bound diagonals} cannot be improved further.
\end{example}

\smallskip

\section{Walks}\label{sec:walks}

The key ingredient in the fact that diagonals may have a big minimal polynomial was the possibility to write them as a sum of residues. The same exponential growth as in Proposition~\ref{prop:generic} therefore occurs for other functions bearing this same structure. For instance, constant terms of rational functions in ${\mathbb C}(x)[[y]]$ can also be written as contour integrals of rational functions around the origin and thus by the residue theorem be expressed as a sum of residues.

By contrast, such sums of residues of rational functions always satisfy a differential equation of only polynomial size~\cite{BoChChLi10}.
Thus, when an algebraic function appears to be connected to a sum of residues of a rational function, the use of this differential structure is much more adapted to the computation of series expansions, instead of going through a potentially large polynomial. 

As an example where this phenomenon occurs naturally, we consider here the enumeration of unidimensional lattice walks, following Banderier and
Flajolet~\cite{BanderierFlajolet2002} and
Bousquet-M\'elou~\cite{Bousquet2006}. Our goal in this section is to study,
from the algorithmic perspective, the series expansions of various generating
functions (for bridges, excursions, meanders) that have been identified as algebraic~\cite{BanderierFlajolet2002}.
One of
our contributions is to point out that although algebraic series
can be expanded
fast~\cite{ChudnovskyChudnovsky1986,ChudnovskyChudnovsky1987a,BostanChyzakLecerfSalvySchost2007},
the precomputation of a polynomial equation could have prohibitive cost.
We overcome this
difficulty by precomputing differential (instead of polynomial) equations
that have polynomial size only, and using them to compute series expansions to
precision~$N$ for bridges, excursions and meanders in time quasi-linear
in~$N$.

\subsection{Preliminaries}

We start with some vocabulary on lattice walks.
A \emph{simple step} is a vector $(1, u)$ with $u\in\ZZ$. 
A \emph{step set} $S$ is a finite set of simple steps.
A \emph{unidimensional walk} in the plane $\mathbb{Z}^2$ built from  $S$ is a finite sequence $(A_0,A_1,\dots,A_n)$ of points in $\mathbb{Z}^2$, such that $A_0 = (0,0)$ and $\overrightarrow{A_{k-1}A_k} = (1, u_k)$ with $(1, u_k)\in S$. In this case $n$ is called the \emph{length} of the walk, and $S$ is the \emph{step set} of the walk. The $y$-coordinate of the endpoint $A_n$, namely $\sum_{i=1}^{n}{y_i}$, is called the final altitude of the walk.
The characteristic polynomial of the step set $S$ is
$$\Gamma_S(y) = \sum_{(1,u)\in S}{y^u}.$$

\begin{figure}
\begin{minipage}{8.5cm}
\rule[.3cm]{0cm}{0cm}
\centerline{\includegraphics[scale=0.41]{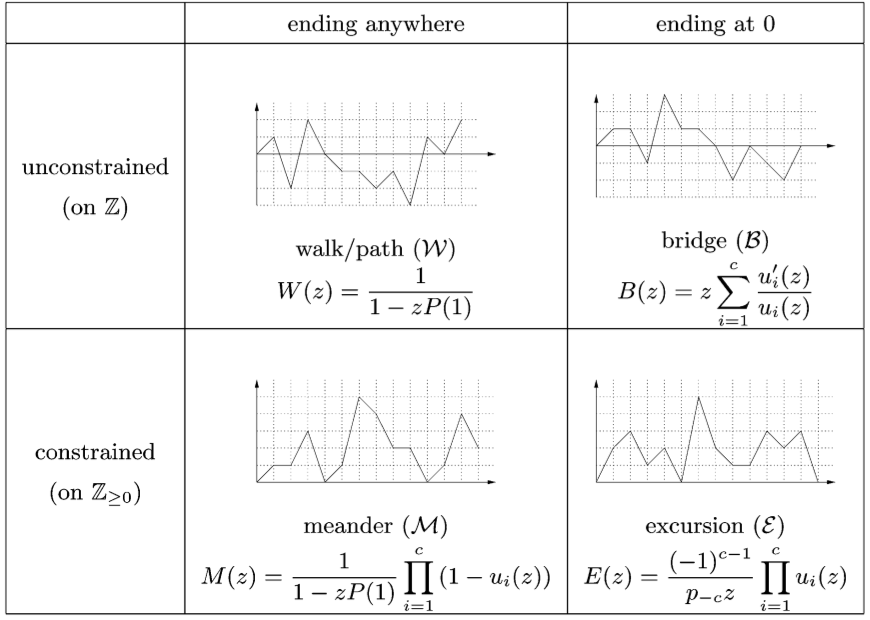}}
\end{minipage}
\caption{\cite{BanderierFlajolet2002} The four types of paths: walks, bridges, meanders and excursions\\ and the corresponding generating functions.
} 
\label{fig:BaFl02}
\vskip-10pt
\end{figure}
Following Banderier and Flajolet, we consider three specific families of walks: bridges, excursions and meanders~\cite{BanderierFlajolet2002}. \emph{Bridges} are walks with final altitude~$0$, \emph{meanders} are walks confined to the upper half plane, and \emph{excursions} are bridges that are also meanders. 
Figure~\ref{fig:BaFl02}, taken from~\cite{BanderierFlajolet2002}, summarizes these definitions graphically.

We define the full generating power series of walks \[W_S(x,y) = \sum_{n\ge 0,k\in\ZZ}{w_{n,k}x^ny^k} \; \in \mathbb{Z}[y,y^{-1}][[x]],\] where $w_{n,k}$ is the number of walks with step set~$S$, of length $n$ and final altitude $k$. We denote by $B_S(x)$ (resp. $E_S(x)$, and $M_S(x)$) the power series $\sum_{n\ge0}{u_{n}x^n}$, where $u_n$ is the number of bridges (resp. excursions, and meanders) of length~$n$ with step set~$S$.

We omit the step set~$S$ as a subscript when there is no ambiguity. Several properties of the power series $W$, $B$, $E$ and $M$ are classical:
\begin{fact}\cite[\S2.1-2.2]{BanderierFlajolet2002}\label{fact:walks}\ The power series $W$, $B$, $E$ and $M$ satisfy
	\begin{compactenum}
	\item[(1)] $W(x,y)$ is rational and $W(x,y) =1/(1-x\Gamma(y))$;
	\item[(2)] $B(x)$, $E(x)$ and $M(x)$ are algebraic;
	\item[(3)] $B(x)=[y^0]W(x,y)$;
	\item[(4)] $E(x) = \exp\left(\int{(B(x)-1)/x\,\mathrm{d}x}\right)$.
	\end{compactenum}
\end{fact}
In what follows, we describe and analyze three methods to compute the power series expansions of $B$, $E$ and $M$. In the next two sections, we first study two previously known methods, then we introduce a new one.

\subsection{Expanding the generating power series}

From now on, we fix a step set $S$, and we denote by $u^-$ (resp. $u^+$) the largest $u$ such that $(1,-u)\in S$ (resp. $(1,u)\in S$). We also define $d=u^- + u^+$. The integer $d$ measures the vertical amplitude of $S$; this makes $d$ a good scale for measuring the complexity of the algorithms that will follow. We assume that both $u^-$ and $u^+$ are positive, since otherwise the study of the bridges, excursions and meanders becomes trivial. 

\smallskip \noindent \textbf{The direct method.}
The combinatorial definition of walks yields a recurrence relation for $w_{n,k}$:
\begin{equation}
\label{eq:rec}
 	w_{n, k} = \sum_{(1, u)\in S}{w_{n-1, k-u}},
\end{equation}
with initial conditions $w_{n,k} = 0$ if $n,k \le 0$ with $(n,k)\neq(0,0)$, and $w_{0,0} = 1$. If $\tilde{w}_{n,k}$ denotes the number of walks of length $n$ and final altitude $k$ that never exit the upper half plane, then $\tilde{w}_{n,k}$ also satisfies recurrence~\eqref{eq:rec}, but with the additional initial conditions $\tilde{w}_{n,k} = 0$ for all $k<0$. Then the bridges (resp. excursions, meanders) are counted by the numbers $w_{n,0}$ (resp. $\tilde{w}_{n,0}$, $\sum_k\tilde{w}_{n,k}$).

One can compute these numbers by unrolling the recurrence relation~\eqref{eq:rec}. Each use of the recurrence costs $\bigO(d)$ ops., and in the worst case one has to compute $\bigO(dN^2)$ terms of the sequence (for example, if the step set is $S = \{ (1,1), \dots, (1,d)\} $). This leads to the computation of each of the generating series in $\bigO(d^2N^2)$ ops. 

This quadratic complexity in $N$ is unsatisfactory, and any method that requires the complete expansion of the generating series $W(x,y)$ is bound to be quadratic in $N$. The two other methods that we are going to present are designed to achieve linear or quasi-linear complexity in $N$. As will be explained, this comes at the cost of a precomputation that must be taken into account in the analysis.

\smallskip \noindent \textbf{Using algebraic equations.}
In \cite[\S 2.3]{BanderierFlajolet2002}, a method relying on the algebraicity of $B$, $E$ and $M$ (Fact~\ref{fact:walks}{\em(2)})) is suggested. 
The series $E$ and $M$ can be expressed as products in terms of the small branches of the characteristic polynomial $\Gamma_S$ (see \cite[Th. 1, Cor. 1]{BanderierFlajolet2002}). From there, a polynomial equation can be obtained using the Platypus algorithm~\cite[\S2.3]{BanderierFlajolet2002}, which computes a polynomial canceling the products of a fixed number of roots of a given polynomial.  Given a polynomial equation $P(z,E)=0$, another one for~$B$ can be deduced from the relation~$B=zE'/E+1$ as $\Resultant_E((B-1)EP_E+zP_z,P)$.

Once a polynomial equation is known for one of these three series, it can be used to compute a linear recurrence with polynomial coefficients satisfied by its coefficients~\cite{ChudnovskyChudnovsky1986,ChudnovskyChudnovsky1987a,BostanChyzakLecerfSalvySchost2007}. The naive algorithm introduced above provides a way to compute a sufficiently large number of initial conditions to unroll this recurrence. (For a quantitative result on the required number of initial conditions, see Corollary~\ref{coro:Expand algebraic functions} below.) This method produces an algorithm that computes the first~$N$ terms of~$B$, $E$ and~$M$ in $\bigO(N)$ ops. For this to be an improvement over the naive method for large~$N$, the dependence on~$d$ of the constant in the $\bigO()$ should not be too large and the precomputation not too costly. 

Indeed, the cost of the precomputation of an algebraic equation is not negligible. The bound $\binom{d}{u^-}$ on the degrees of equations for excursions has been obtained by Bousquet-M\'elou, and showed to be tight for a specific family of step sets, as well as generically~\cite[\S2.1]{Bousquet2006}. This bound may be exponentially large with respect to $d$. Empirically, the polynomials for $B$ and $M$ are similarly large. 

The situation for differential equations and recurrences is different: $B$ satisfies a differential equation of only polynomial size (see below), whereas (empirically), those for $E$ and $M$ have a potentially exponential size. These sizes then transfer to the corresponding recurrences and thereby to the constant in the complexity of unrolling them.
The purpose of Theorem~\ref{thm:walks} below is to give explicitly the polynomial dependence in $d$ when using this method, showing at the same time that a true improvement over the naive method can be achieved. 
\begin{example} With the step set $S = \left\{(1,d), (1,1), (1,-d)\right\}$ and $d\ge 2$, the counting series $W_{S}$ equals
\[W_S(x,y) = \frac{y^d}{y^d-x(1+y^{d+1}+y^{2d})}.\]
Experiments indicate that the minimal polynomial of $B_S(x)$ has bidegree $(2d\binom{2d-2}{d-1}, \binom{2d}{d})$, exhibiting an exponential growth in~$d$. 
On the other hand, they show that $B_S(x)$ satisfies a linear differential equation of order $2d-1$ and coefficients of degree $d^2+3d-2$ for even~$d$, and $d^2+3d-4$ for odd $d$.
\end{example}

\smallskip \noindent {\bf New Method.}
We now give a method that runs in quasi-linear time (with respect to $N$) and avoids the computation of an algebraic equation. Our method relies on the fact that periods of rational functions such as the one in Part~{\em(3)} of Fact \ref{fact:walks} satisfy differential equations of polynomial size in the degree of the input rational function~\cite{BoChChLi10}.
We summarize our results in the following theorem, and then 
go over the proof in each case individually.

\begin{thm}\label{thm:walks} Let $S$ be a finite set of simple steps and $d=u^-+u^+$.
	The series $B_S$ (resp. $E_S$ and $M_S$) can be expanded at order $N$ in $\bigO(d^2N)$ ops. (resp. $\softO(d^2N)$ ops.), after a precomputation in $\softO(d^5)$~ops.
\end{thm}
\subsection{Fast Algorithms}

\noindent \textbf{Bridges.} To expand $B(x)$, we rely on Fact~\ref{fact:walks}{\em(3)}. The formula can be written $B=(1/2\pi i)\oint{W(x,y)\frac{dy}{y}}$, the integration path being a circle inside a small annulus around the origin~\cite[proof of Th.~1]{BanderierFlajolet2002}. Moreover, $W(x,y)/y$ is of the form $P/Q$, where $\bideg Q \le (1, d)$ and $\bideg P \le (0, d-1)$. Since $P$ and $Q$ are relatively prime and $Q$ is primitive with respect to $y$, Algorithm \textbf{HermiteTelescoping}~\cite[Fig.~3]{BoChChLi10}  computes a telescoper for $P/Q$, which is also a differential equation satisfied by $B$. By Fact~\ref{fact:sqfree}{\em(2)}, the resulting differential equation has order at most~$d$ and degree $O(d^2)$, and is computed using $\softO(d^5)$ ops.
This differential equation can be turned into a recurrence of order $r = O(d^2)$ in quasi-optimal time (see the discussion after~\cite[Cor. 2]{BoSc05}). We may use it to expand $B(x) \bmod x^N$ in $O(d^2N)$ ops, once enough initial conditions are known. Again, the initial conditions are computed by means of the direct method. The only remaining question is the number of initial conditions needed. Indeed, the recurrence may be singular, ie its leading coefficient may have positive integer roots. If we denote by $\alpha$ the largest such root, then we need to compute the first terms of the recurrence up to $\max(r-1,\alpha)$. In order not to break the flow of reading, we postpone the discussion on the size of $\alpha$ to the next section. For now, we only state the result.
\begin{prop}\label{prop:initial conditions}
	Let $S$ be a set of simple steps, and $d = \max_{(1,u), (1,v)\in S}{|u-v|}$.
	Then the largest integer root of the leading term of the recurrence computed by Algorithm~\ref{algo:bridges} is at most $\bigO(d^3)$
\end{prop}
\begin{proof}
	See Section~$\ref{subsec:singular recurrences}$.
\end{proof}

Thus, a sufficient number of initial conditions is computed with $\bigO(d^5)$ ops by the direct method, and the total cost of the precomputation is $\softO(d^5)$, as announced.

\begin{algo}
	Algorithm \textbf{Walks}($S$, $N$)
	
	\begin{algoenv}{A set $S$ of simple steps and an integer $N$}{$B_S,E_S,M_S\bmod x^{N+1}$}
		\State $F \gets W(x,y)/y$ [case $B, E$] or $W(x,y)/(1-y)$ [case $M$]
		\State $D\gets \textbf{HermiteTelescoping}(F)$ \cite[Fig.~3]{BoChChLi10}
		\State $R\gets $ the recurrence of order $r$ associated to $D$
		\State $I\gets [y^0]W(x,y)\bmod x^{r+1}$ [case $B,E$]\\ \qquad $[y^0]yW(x,y)/(1-y)\bmod x^{r+1}$ [case $M$]
		\State $B \gets  [y^0]W(x,y)\bmod x^{N+1}$ (from $R,I$)
		\State $A \gets  [y^0]yW(x,y)/(1-y)\bmod x^{N+1}$ (from $R,I$)
		\State $E\gets \exp\left(\int{(B(x)-1)/x\,\mathrm{d}x}\right)\bmod x^{N+1}$
		\State $M\gets \exp\left(-\int{(A(x)/x)/(1-\Gamma(1)x)}\,\mathrm{d}x\right)\bmod x^{N+1}$
		\State\Return$B,E,M$
	\end{algoenv}
	\caption{Expanding the generating functions of bridges, excursions and meanders}\label{algo:bridges}
\end{algo}

\smallskip \noindent \textbf{Excursions.} If $B(x) \bmod x^{N+1}$ is known, it is then possible to recover $E(x) \bmod x^{N+1}$ thanks to Fact~\ref{fact:walks}{\em(4)}. Expanding $E(x)$ comes down to the computation of the exponential of a series, which can be performed using $\softO(N)$ ops. (Fact~\ref{fact:complexity}{\em(4)}).

\smallskip \noindent \textbf{Meanders.} As in the case of excursions, the logarithmic derivative of $M(x)$ is recovered from a sum of residues by the following.
\begin{prop} \label{walks-formulas2} The  series~$W$ and~$M$ are related through
\[A(x) = [y^0]\frac{y}{1-y}W(x,y),\quad M(x) = \frac{\exp\left(-\int\frac{A(x)}{x}\,\mathrm{d}x\right)}{1-x\Gamma(1)}.\]
\end{prop}
\begin{proof}
	Denote by $y_1,\dots,y_{u^-}$ the small branches of the polynomial $y^{u^-} - xy^{u^-}\Gamma(y)$. Then $M$ is given as~\cite[Cor.~1]{BanderierFlajolet2002}:
  \[M(x)=\frac{1}{1-x\Gamma(1)}\prod_{i=1}^{u^{-}}{(1-y_i)}.\]
  On the other hand,
\begin{align*}
    A(x) &= \frac1{2\pi i} \oint{\frac{W(x,y)}{1-y}\mathrm{d}y}\\
    &= \sum_{i=1}^{u^-}{\Residue_{y=y_i(x)}\left(\frac{1}{(1-y)(1-x\Gamma(y))}\right)} 
    =  -\sum_{i=1}^{u^-}{\frac{1}{(1-y_i)x\Gamma '(y_i)}},
\end{align*}
  where the integral has been taken over a circle 
  around the origin and the small branches.
	Differentiating the equation $1-x\Gamma(y)$ = 0 with respect to $x$ leads to
$-x\Gamma '(y_i) = {1}/({xy_i'})$, whence
	$A(x) = x\sum_{i=1}^{u^{-}}{{y_i'}/({1-y_i})}.$ Therefore, $\prod ( 1- y_i) = \exp ( - \int A/x\,\mathrm{d}x) )$, finishing the proof.
\end{proof}
Thus we apply the same method as in the case of the excursions. We first compute a differential equation for $A(x)$ using the method of~\cite{BoChChLi10}. The computation of the initial conditions for $A$ can also be performed naively from its definition as a constant term, by simply expanding $yW(x,y)/(1-y)$. The formula of the proposition then recovers $M(x)$. The complexity analysis goes exactly as in the previous case, giving a global cost of $\softO(d^5)$ ops.

\subsection{Singular recurrences}\label{subsec:singular recurrences}

We now come back to the problem of singular recurrences. In our context, the recurrences that we come across have a very specific structure: they are associated to differential resolvents of polynomials. (The differential resolvent of a polynomial is the least order differential operator canceling all of its roots.) This structure can be exploited to derive bounds on the singularities of our recurrences.

If $P\in\KK[x][y]$ is a polynomial, consider the recurrence associated to its differential resolvent $L$. The leading coefficient of this recurrence is called the indicial polynomial of $L$ at~0. Its largest integer root will be denoted $\alpha$. The fundamental idea is that there exists a Laurent series solution of $L$ which has valuation $\alpha$~\cite[\S 15.31]{Ince1956}
Therefore, it is sufficient to find bounds on the valuations of the solutions of $L$. This is done in the following theorem.

\begin{thm}\label{th:valuations}
	Let $P$ be a polynomial in $\KK[x][y]$, of bidegree at most $(d_x, d_y)$, and $L$ be the differential resolvent of $P$.
	Then all the Laurent series solutions $y(x)$ of $L$ uniformly satisfy
	\[\val_x(y(x)) = \bigO(d_xd_y^2).\]
\end{thm}

\begin{proof}
	Choose a subfamily $y_1,y_2,\dots,y_n$ of the Puiseux series roots of $P$ that constitutes a basis of the solution space of the resolvent (in particular, $n\le d_y$).
	Let $y=\sum_{i=1}^n{\lambda_iy_i}$ be a Laurent series solution of the differential resolvent of $P$. 

	Then the fact that $\val(f') \ge \val(f)-1$ for any Laurent series $f\in\KK[[x]]$ implies that
	\[\val(\Wr(y,y_2,\dots,y_n)) \ge \val(y)+\sum_{i=2}^n{\val(y_i)} - \binom{n}{2}.
	\]
	By the multilinearity of the Wronskian, the left-hand side of this inequality is nothing more than $\val(\Wr(y_1,y_2,\dots,y_n))$. On the other hand, the absolute values of the valuations of the $y_i$'s are bounded by $\max(d_x, d_y)$ (because they are slopes of edges in the Newton polygon of $P$). A bound for $\val(y)$ is thus obtained:
	\[\val(y)\le \val(\Wr(y_1,y_2,\dots,y_n)) + (d_y-1)\max(d_x,d_y) + \frac{d_y(d_y-1)}{2}.\] 
	The proof is then reduced to showing that $\val(\Wr(y_1,y_2,\dots,y_n)) = \bigO(d_xd_y^2)$.
	This is very similar to the computations conducted in \cite[\S 2.2]{BostanChyzakLecerfSalvySchost2007}. We start by recalling some facts that are proved there. There exist polynomials $W_k\in\KK[x,y]$ such that for all $i\in\{1,2,\dots,n\}$ and all $k\ge 1$, the derivative $y_i^{(k)}$ can be expressed as
	\[y_i^{(k)}=\frac{W_k(x,y_i)}{P_y(x,y_i)^{2k-1}}.\]
	Moreover, the polynomials $W_k$ satisfy
	\begin{equation}
	\deg_x W_k \le (2d_x-1)k-d_x, \quad \deg_y W_k\le 2(d_y-1)k-d_y+2.
	\end{equation}

	It follows that $D = \prod_{i=1}^n{P_y(x,y_i)^{2n-3}}\in\KK[x,y_1,y_2,\dots,y_n]$ is a polynomial such that $\Wr(y_1,y_2,\dots,y_n)\cdot D\in\KK[x,y_1,y_2,\dots,y_n]$.
	We will denote by $R$ this last polynomial. $R$ is the determinant of the matrix
	\begin{equation*}
	\mathcal{N}=
	\begin{pmatrix}
	y_1P_y(x,y_1)^{2n-3} & \cdots & y_nP_y(x, y_n)^{2n-3} \\
	W_1(x,y_1)P_y(x, y_1)^{2n-4} & \cdots & W_1(x, y_n)P_y(x, y_n)^{2n-4} \\
	W_2(x,y_1)P_y(x, y_1)^{2n-6} & \cdots & W_2(x, y_n)P_y(x, y_n)^{2n-6} \\
	\vdots & \vdots & \vdots \\
	W_{n-1}(x, y_1) & \cdots & W_{n-1}(x, y_n)
	\end{pmatrix}.
	\end{equation*}
	$R$ is an anti-symmetric polynomial in $y_1,y_2,\dots,y_n$, but $R^2$ is symmetric, as well as $D$, so we can apply Lemma~\ref{lem:symmetric functions} to see that $R^2$ and $D$ belong to $\KK(x)$. Therefore, the equality \[\Wr(y_1,y_2,\dots,y_n) = \frac{R}{D}\]
	shows that $\Wr(y_1,y_2,\dots,y_n)$ is the square root of a rational function in $x$. We are going to use this structure and Lemma~\ref{lem:symmetric functions} to derive the desired bound on the valuation of the Wronskian determinant.
	
	If $\det(\mathcal{N})$ is viewed as a polynomial in $\KK[x,y_1,y_2,\dots,y_n]$, then
	\begin{align*}
		\deg_x \det(\mathcal{N})^2 \le & \ 2 \sum_{k=0}^{n-1}{((2n-3)d_x-k)} \\
		\le & \ 2n(2n-3)d_x + n(n-1),
	\end{align*}
	and for all $i\in\{1,2,\dots,n\}$,
	\[ \deg_{y_i} \det(\mathcal{N}) \le 2(2n-3)d_y-2(2n-4). \]
	Similarly, when $D$ is viewed as a polynomial in $\KK[x,y_1,y_2,\dots,y_n]$, we have:
	\begin{align*}
		\deg_x D = n(2n-3)d_x, \quad \deg_{y_i}D = (2n-3)(d_y-1).
	\end{align*}
	Applying Lemma ~\ref{lem:symmetric functions}, we deduce that, denoting by $p(x)$ the leading coefficient of $P(x,y)$,
	\begin{align*}
		\Wr(y_1,y_2,\dots,y_n) = \frac{U(x)}{p(x)V(x)},
	\end{align*}
	where
	\[\deg_x U^2 \le 2n(2n-3)d_x+n(n-1)+2(2n-3)d_xd_y-2(2n-4)d_x.\]
	Finally, the inequalities $\val(\Wr(y_1,y_2,\dots,y_n))\le \frac{1}{2}\val(U^2)\le \frac{1}{2}\deg_x(U^2)$ and $n\le d_y$ yield
	$$\val(\Wr(y_1, y_2,\dots,y_n)) = \bigO(d_xd_y^2),$$
	which concludes the proof.	
\end{proof} 

We immediately deduce the following corollary on the number of initial conditions required to expand an algebraic power series.

\begin{coro}\label{coro:Expand algebraic functions}
	Let $P\in\KK[x,y]$ be a polynomial of bidegree bounded by $(d_x, d_y)$. Let $R$ be the recurrence associated to the differential resolvent of $P$.
	Then the largest integer root of the leading coefficient of $R$ is at most $\bigO(d_xd_y^2)$.
\end{coro}

\begin{proof}
	Immediate from the theorem and the discussion that precedes it.
\end{proof}

We are now able to prove Proposition~\ref{prop:initial conditions}.

\begin{proof} (of Proposition~\ref{prop:initial conditions})
	We only treat the case where the recurrence is computed for $B$, and the proof transposes directly to the case of $A$.
	Let $S$ and $d$ be as in the Proposition, and denote by $P$ the minimal polynomial of $B$. Then the recurrence computed by Algorithm~\ref{algo:bridges} is associated to the minimal annihilating differential operator for $B$, which is also the differential resolvent of $P$. We denote it by $L_P$. Now since $B = [y^{-1}]W(x,y)/y$, it can be written as a sum of residues similar to formula~(\ref{eq:diagonal-as-residues}). If we denote by $R$ the polynomial that cancels these residues, then $P$ divides $\Sigma_c R$ for some $c$. This implies in particular that all the solutions of $L_P$ are linear combinations of the roots of $R$. Thus, if $L_R$ is the differential resolvent of $R$, then all the solutions of $L_P$ are solutions of $L_R$. Since $W$ has bidegree $(1,d)$, Theorem~\ref{th:Bronstein} and Theorem~\ref{th:valuations} show that all the roots of $P$ have valuation at most $\bigO(d^3)$, and the result follows.
\end{proof}

\section{Conclusion}
We gave a complete and efficient algorithm that calculates a polynomial equation satisfied by the diagonal of a bivariate rational function in  characteristic~$0$. Generically, the degree in $\Delta$ of the polynomial  $P(t,\Delta)$ output by the algorithm is optimal. The bound on the degree in~$t$ is not tight. The gap between this bound and the actual degrees is not yet fully understood: it is already present for the Rothstein-Trager and Bronstein resultants. 
Our complexity results are given in the arithmetic complexity model. The
corresponding study in the binary model remains to be done.

The case of positive characteristic requires different methods and algorithms.
In that case, diagonals are algebraic even for rational functions with more than two variables. To the best of our knowledge, these questions have never been studied from the complexity viewpoint. One possible direction is to try and make effective the proof by Furstenberg that these diagonals are algebraic~\cite{Furstenberg1967}. Some work has also been done by Adamczewski and Bell~\cite{AB13} who among other things studied how the sizes of the polynomial equations satisfied by diagonals vary with the characteristic of the base field.

\bibliographystyle{abbrv}

\end{document}